\documentclass{article}

\usepackage{PRIMEarxiv}

\usepackage{times}
\usepackage{helvet}
\usepackage{courier}
\usepackage{url}
\usepackage{caption}
\usepackage{graphicx} 
\usepackage{hyperref}
\usepackage{comment}

\usepackage[ruled,vlined,linesnumbered]{algorithm2e}
\usepackage{comment}

\usepackage{environ}
\newif\ifshort

\ifshort
	\makeatletter
	\newenvironment{claimO}[1][]{
		\if\relax\detokenize{#1}\relax
		\expandafter\@firstoftwo
		\else
		\expandafter\@secondoftwo
		\fi
		{\begin{claim}[$\star$]}{\begin{claim}[#1, $\star$]}
			}
			{
			\end{claim}
		}
		\makeatother
		
	\makeatletter
	\newenvironment{lemmaO}[1][]{
		\if\relax\detokenize{#1}\relax
		\expandafter\@firstoftwo
		\else
		\expandafter\@secondoftwo
		\fi
		{\begin{lemma}[$\star$]}{\begin{lemma}[#1, $\star$]}
			}
			{
			\end{lemma}
		}
		\makeatother
		
		\makeatletter
		\newenvironment{propositionO}[1][]{
			\if\relax\detokenize{#1}\relax
			\expandafter\@firstoftwo
			\else
			\expandafter\@secondoftwo
			\fi
			{\begin{proposition}[$\star$]}{\begin{proposition}[#1, $\star$]}
				}
				{
				\end{proposition}
			}
			\makeatother
			
			\makeatletter
			\newenvironment{observationO}[1][]{
				\if\relax\detokenize{#1}\relax
				\expandafter\@firstoftwo
				\else
				\expandafter\@secondoftwo
				\fi
				{\begin{observation}[$\star$]}{\begin{observation}[#1, $\star$]}
					}
					{
					\end{observation}
				}
				\makeatother			
				
				\makeatletter
				\newenvironment{theoremO}[1][]{
					\if\relax\detokenize{#1}\relax
					\expandafter\@firstoftwo
					\else
					\expandafter\@secondoftwo
					\fi
					{\begin{theorem}[$\star$]}{\begin{theorem}[#1, $\star$]}
						}
						{
						\end{theorem}
					}
					\makeatother
					\makeatletter
					\newenvironment{corollaryO}[1][]{
						\if\relax\detokenize{#1}\relax
						\expandafter\@firstoftwo
						\else
						\expandafter\@secondoftwo
						\fi
						{\begin{corollary}[$\star$]}{\begin{corollary}[#1, $\star$]}
							}
							{
							\end{corollary}
						}
						\makeatother
						\NewEnviron{proofO}{}
						\NewEnviron{claimproofO}{}
\else

						\makeatletter
						\newenvironment{proofO}[1][Proof]{\begin{proof}[#1]}{\end{proof}}
						\makeatother
						\makeatletter
						\newenvironment{claimproofO}[1]{\begin{claimproof}[#1]}{\end{claimproof}}
						\makeatother

\fi

\usepackage{newfloat}
\usepackage{listings}
\usepackage{amsthm}
\usepackage{amssymb}
\usepackage{mathtools}
\usepackage{diagbox}
\usepackage{multirow}
\usepackage{graphicx,color}

\newenvironment{lcases}
  {\left\lbrace\begin{aligned}}
  {\end{aligned}\right.}

\usepackage{xspace}
\usepackage{cleveref}
\usepackage{todonotes}
\usepackage{thm-restate}

\newtheorem{theorem}{Theorem}
\newtheorem{lemma}{Lemma}
\newtheorem{claim}{Claim}
\newtheorem{observation}{Observation}
\theoremstyle{definition}
\newtheorem{definition}{Definition}
\theoremstyle{definition}
\newtheorem{question}{Open Question}

\newcommand{\cqed}{\renewcommand{\qedsymbol}{$\lrcorner$}\qed}
\newenvironment{claimproof}{\noindent \emph{Proof of Claim~\theclaim.}}{\hfill\cqed\medskip}

\DeclareCaptionStyle{ruled}{labelfont=normalfont,labelsep=colon,strut=off} 
\newcommand{\EF}{\textsc{EF}\xspace}
\newcommand{\EFone}{\textsc{EF1}\xspace}

\newcommand{\EFX}{\textsc{EFX}\xspace}

\newcommand{\EFXandPO}{\textsc{EFX+PO}\xspace}
\newcommand{\PO}{\textsc{PO}\xspace}
\newcommand{\fPO}{f\textsc{PO}\xspace}
\newcommand{\MNW}{\textsc{MNW}\xspace}

\title{EFX and PO Allocation Exists for Two Types of Goods}
\author{
  Vladimir Davidiuk\\
  St. Petersburg State University\\
  \And
  Yuriy Dementiev\\
  ITMO University\\
  \And
  Artur Ignatiev\\
  ITMO University\\
  \And
  Danil Sagunov\\
  ITMO University\\
}
\date{}

\begin{document}

\maketitle

\begin{abstract}


We study the problem of fairly and efficiently allocating indivisible goods among agents with additive valuations.
We focus on envy-freeness up to any good (\EFX)---an important fairness notion in fair division of indivisible goods. A central open question in this field is whether \EFX allocations always exist for any number of agents. While prior work has established \EFX existence for settings with at most three distinct valuations \cite{three_types_agents_EC} and for two types of goods \cite{two_types_goods}, the general case remains unresolved.

In this paper, we extend the existent knowledge by proving that \EFX allocations satisfying Pareto optimality (\PO) always exist and can be computed in quasiliniear time when there are two types of goods, given that the valuations are positive. This result strengthens the existing work of \cite{two_types_goods}, which only guarantees the existence of \EFX allocations without ensuring Pareto optimality.
Our findings demonstrate a fairly simple and efficient algorithm constructing an \EFXandPO allocation.

\end{abstract}

\section{Introduction}
Fair division of indivisible goods is a core research area in algorithmic game theory and computational social choice, focusing on the equitable allocation of discrete items—such as property, licenses, or humanitarian aid—among agents with heterogeneous preferences. 
The online platform Spliddit (spliddit.org) offers practical implementations of fair division algorithms, addressing real-world allocation challenges including rent distribution among roommates, taxi fare splitting, and fair assignment of goods between individuals.
A central goal is to achieve fairness and effectiveness guarantees that balance efficiency and equity, even when exact solutions are theoretically or computationally out of reach.

The gold standard of fairness, envy-freeness (EF) \cite{EF_foley1966resource}, ensures no agent prefers another’s allocation over their own.
Yet, EF allocations often fail to exist for indivisible goods, and even when they do, identifying them is NP-complete as shown by \cite{lipton} even with two agents.
This limitation has spurred the study of meaningful relaxations. 
The initial relaxation of envy-freeness is the concept of envy-freeness up to one good (\EFone), which was informally introduced by \cite{lipton} and later formally defined by \cite{Budish}. Under \EFone, an agent $i$ may envy agent $j$, provided there is at least one good in $j$’s bundle such that, if removed, $i$’s envy toward $j$ would disappear.
\EFone allocations always exist for indivisible goods under additive valuations and can be computed in polynomial time using algorithms like the envy cycle elimination method \cite{lipton}.
\cite{Caragiannis_2019} proved that Round-Robin algorithms leads to an \EFone allocation.
As proven by \cite{Caragiannis_2019}, allocations maximizing the Nash Social Welfare (MNW) are necessarily Pareto-efficient and envy-free up to one good.
Among these, envy-freeness up to any good (\EFX), introduced by \cite{Caragiannis_2019}, has emerged as a compelling alternative. 
\EFX softens the strictness of EF: while an agent may envy another’s bundle, this envy vanishes upon the removal of any single good from the envied bundle.
The existence of \EFX is still open for $n$ agents with arbitrary additive valuations.

Recent years have seen significant progress in answering this question, though a general solution remains elusive. 
The first positive existence result was shown by \cite{PlautR20} for when the agents have identical valuations, they showed that the lexicographic minimal (lexmin) solution guarantees \EFX in this case.
For two agents, \cite{PlautR20} showed that the cut and choose algorithm guarantees \EFX. 
They also extended their analysis beyond identical valuations, demonstrating that the Envy-Cycle Elimination algorithm also guarantees \EFX allocations for ordered instances---cases where agents share identical preference rankings over goods but may differ in their cardinal valuations.
\cite{ChaudhuryGM24_3agent_EFX} extended this result to three agents, while  \cite{two_types_agents} and \cite{three_types_agents_EC} proved existence for settings with two and three types of agents, respectively. 
On the other side, \cite{two_types_goods} demonstrated that \EFX allocations exist when goods belong to just two distinct types, even providing a partial characterization of when exact envy-freeness (\EF) is achievable in such cases.
A previous paper established the existence of EFX and PO allocations under lexicographic preferences \cite{EFX_lexicographic}.
Moreover, \cite{EFX_bivalued} proved that an EFX and fPO allocation exists and can be computed in polynomial time for bivalued instances.

Pareto optimality is a fundamental efficiency concept in fair division, ensuring no agent can improve their allocation without harming others.
\cite{mnw_efx} showed that when goods take on at most two possible values, the Maximum Nash Welfare (\MNW) solution guarantees \EFX and Pareto optimality (\PO), also they show a polynomial-time algorithm for \EFX based on matching. 
However, \cite{PlautR20} also revealed a sobering limitation: if agents can assign zero value to goods, \EFX and \PO may be incompatible. This leaves an open problem: Does an \EFX and \PO allocation always exist when all valuations are positive?

\paragraph{Our Contribution.}

We significantly advance the understanding of fair and efficient allocations for two types of goods. First, we prove that EFX+PO allocation always exist when all agents have positive utilities, while demonstrating that the stronger EFX+fPO combination remains impossible. This improves upon \cite{two_types_goods}, who established EFX existence without Pareto optimality.

We provide an efficient algorithmic solution: our approach computes such an allocation in $\mathcal{O}(n \log n + \log m)$ time for $n$ agents and $m$ goods.
It works as fast as $\mathcal{O}(\log n + \log m)$ time, if the agents are ordered by their relative valuations of good types. Additionally, we introduce a class of \emph{proper} allocations and prove that every \emph{proper} allocation is Pareto optimal, offering new structural insights into the interplay between fairness and efficiency.

\paragraph{Example: EFX and fPO are not always compatible.} 
Consider a simple example with two types of goods (croissant and coffee) and two agents (Alice ($1$) and Bob ($2$)). There are two croissants ($a$) and two coffees ($b$), with valuations $v_1(a) = v_2(a) = 1$ and $v_1(b) = 10$ and $v_2(b) = 9$. 
In the unique \EFX allocation, each agent receives one croissant and one coffee.
While this allocation is envy-free up to any good, it fails to be fractionally Pareto optimal (fPO): a fractional allocation where Alice gets $1.1$ coffee and Bob gets $0.9$ coffee plus two croissant would strictly improve the allocation, Alice has the same utility, but Bob increases his utility ($11$ for Alice and $10.1$ for Bob).
However, in the indivisible setting, this EFX allocation is Pareto optimal (PO), as no discrete reallocation can improve one agent's utility without harming the other.
This demonstrates that there may not be allocations that satisfy \EFX and \fPO simultaneously.
In our turn, we show that an \EFXandPO allocation always exists.

\begin{table}[ht]
    \centering
    \includegraphics[scale=0.3]{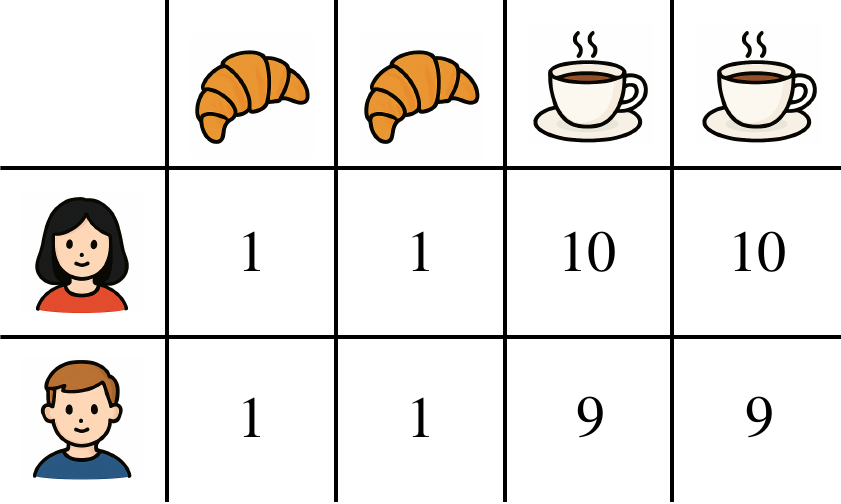}
    \caption{Example with two types of goods.}\label{fig:example}
\end{table}

\paragraph{Additional Related Work.}
The fair division of indivisible goods has witnessed remarkable theoretical and algorithmic progress in recent years, as documented in several comprehensive surveys \cite{overview23, aziz2022algorithmic,complexity_results_survey}.

One promising direction involves \EFX relaxation by allowing a small subset of goods to remain unassigned---often interpreted as charitable donations---while maintaining fairness guarantees.
While trivially satisfying envy-freeness by leaving all goods unallocated is meaningless, meaningful progress has been made in bounding both the number of discarded items and their welfare impact.
\cite{Gravin_donating} demonstrated that an \EFX allocation exists for a subset of goods while preserving at least half of the Maximum Nash Welfare.
\cite{ChaudhuryKMS21} developed an algorithm for computing partial \EFX allocations where: (i) at most $n-1$ goods remain unassigned, and (ii) no agent strictly prefers the set of unallocated goods to their own bundle, while \cite{BergerCFF22} and \cite{Mahara24} later reduced this to $n-2$ goods in general and just one good for four agent instances.
\cite{GhosalHN025} generalized these results, proving that for agents with at most $k$ distinct valuations, \EFX allocations exist with $k-2$ unassigned goods.

Approximate \EFX notions have also gained attention, where an allocation is $\alpha$-EFX if no agent envies another after scaling the other’s bundle without any good by $\alpha$, formally if for every pair of agents $i$ and $j$: $v_i(X_i) \geq \alpha \cdot v_i(X_j\setminus g)$ for any $g \in X_j$.
\cite{PlautR20} established the existence of $0.5$-\EFX allocations, and this bound was later improved to $0.618$ by \cite{AmanatidisMN20}. Most recently, \cite{AmanatidisFS24} showed that $\frac{2}{3}$-EFX allocations exist for up to seven agents or when agents have no more than three distinct valuations.
Building on this line of work, \cite{ChaudhuryGMMM21} elegantly bridged the concepts of approximation and charity, demonstrating that a $(1-\varepsilon)$-\EFX allocation can be efficiently computed while donating only a sublinear number of goods, all while preserving high Nash welfare for $\varepsilon \in (0, 0.5]$.

\paragraph{Setting.}
The setting consists of a set of $n$ agents $N = [n] = \{1, \ldots, n\}$, and a multiset  $M$ of indivisible goods that contains copies of two different goods: $g_1$ with multiplicity $m_1$ and $g_2$ with multiplicity $m_2$. 
Each agent $i \in N$ has a valuation function $v_i\colon 2^M \to \mathbb{R}_{>0}$, which quantifies the utility that $i$ derives from any subset of goods.
We focus on additive valuations, meaning the value of a subset of goods is simply the sum of the values of its individual items. Formally, for any subset $S \subseteq M$, the valuation function satisfies $v_i(S) = \sum_{g\in S} v_i(g)$. 
An allocation $X = (X_1, X_2, \dots, X_n)$ is a partition of $M$ into $n$ disjoint subsets, called bundles, where each agent $i$ receives the bundle $X_{i}$.
We match allocation $X$ with the multiplicity of goods in each agent's bundle $\{(x_{i,1}, x_{i,2}) \mid i\in [n]\}$, so we can assume that agents utility equals $v_i(X_i) = x_{i,1}\cdot v_{i,1} + x_{i,2}\cdot v_{i,2}$, where $v_{i,j}=v_i(\{g_j\})$ for each $i\in [n], j\in [2]$.

\begin{definition}[\EFX]
    An allocation $X$ is envy-free up to any good (\EFX) if, for every pair of agents $i, j \in N$, it holds that $v_i(X_i) + v_i(g) \geq v_i(X_j)$ for any $g\in X_j$.

   We say that agent $i$ envies agent $j$ up any good if $v_i(X_i) + \min_{g \in X_j} v_i(g) < v_i(X_j)$.
\end{definition}

\begin{definition}[\PO]
    An allocation $X$ is Pareto optimal (\PO) if there is no allocation $Y$ such that $v_i(Y_i) \geq v_i(X_i)$ for all $i \in N$ and $v_j(Y_j) > v_j(X_j)$ for some $j\in N$. Equivalently, we will say that such an allocation is not Pareto dominated by any other allocation.
\end{definition}

\section{The Setup}


To show our main result, we have to define several crucial concepts.
In this section, we give necessary definitions and formulate fundamental properties of our allocation structures.
They serve as a backbone of our algorithm.

Most of the properties in this section are given without proofs to ease the understanding of the construction and the algorithm itself.
Complete proofs of such properties and claims can be found in subsequent sections.

\paragraph{Input preprocessing.}
Our first goal is to make the input data much more convenient and reflecting the actual allocation properties. 

By $n$ we denote the number of agents, by $m_1$ and $m_2$ we denote the total quantities of items of first and second types respectively.
We naturally assume that $n\ge 1$ and $m_1,m_2\ge 1$.
Agents are identified with integer numbers in $[n]$.
For $i\in [n]$, $j\in [2]$, let $v_{i,j}$ denote the utility of an item of type $j$ as estimated by agent $i$.

We then divide the agents in two groups, the first group consists of agents $i$ that have $v_{i,2}\le v_{i,1}$ (agents that prefer a good of type $1$ to a good of type $2$ or value them equally),
the second consists of agents $i$ that have $v_{i,2}\ge v_{i,1}$\footnote{We do not use $v_{i,2}>v_{i,1}$ here for purpose.}.
Note that there may be several ways to divide the agents, choose any of them arbitrarily.
Let $n_1, n_2$ denote the obtained group sizes.
If it appears that $\frac{m_1}{n_1}< \frac{m_2}{n_2}$, we interchange the good types and the groups (that is why we defined the groups symetrically).
Then the inequality $\frac{m_1}{n_1} \geq \frac{m_2}{n_2}$ holds. If $n_2 = 0$, we consider that $\frac{m_2}{n_2} = \infty$.
This means, in particular, that $n_2>0$.

We assume that $v_{i,j}>0$ for each $i\in[n], j\in [2]$, as motivated in the introduction.
Then, we normalize the agent utilities: we replace $(v_{i,1},v_{i,2})$ with $(1,v_{i,2}/v_{i,1})$ for each agent $i\in [n]$.
Finally, we re-enumerate the agents in the way that $v_{i,2}\le v_{i+1,2}$ holds for each $i\in[n-1]$.
By definition of $n_1$ and $n_2$, we have that $\forall i \in [n_1]: v_{i,2}\le 1$ and $v_{n_1+1,2}\ge 1$.

The other thing crucial to our further constructions is that the total number of goods is at least the number of agents.

\begin{observation}\label{obs:too-few-items}
    If $m_1+m_2\le n$, give one good of the first type to each agent $i$ with $i\le m_1$.
    Give one good of the second type to each agent $i$ with $i\ge n-m_2+1$.
    This allocation is EFX+PO.
\end{observation}
\begin{proof}
The observed allocation $X$ is simply EFX, since each agent receives at most one good.

To see that the allocation is Pareto-optimal, assume the contrary, and there is an allocation $Y$ that Pareto-dominates $X$.
First, if an agent received one item in $X$ (its utility in $X$ is positive), he should receive at least one item in $Y$.
Hence, each agent gets at most one item in $Y$.
The total utility over all agents in $Y$ equals
$|S_1|+\sum_{i\in S_2} v_{i,2}=m_1+\sum_{i\in S_2} v_{i,2},$
where $S_1$ and $S_2$ are the sets of agent numbers that get a good of type $1$ and type $2$ in $Y$ respectively.

As soon as $Y$ Pareto-dominates $X$, the total utility of all agents in $X$ should be less than $m_1+\sum_{i\in S_2} v_{i,2}$.
This is impossible because $\sum_{i\in S_2}v_{i,2}$ is upper-bounded by $\sum_{i=n-m_2+1}^n v_{i,2}$ --- the total utility provided by goods of the second type in $X$.
\end{proof}

We are ready to summarize our restrictions on the input, that all hold true after the input is preprocessed.
We provide them below.

\paragraph{Preprocessed input restrictions}
\begin{itemize}
    \item For each $i\in[n]$, $v_{i,1}=1$ and $v_{i,2}>0$;
    \item For each $i\in [n-1]$, $v_{i,2}\le v_{i+1,2}$;
    \item $n_1+n_2=n$ and $v_{n_1,2}\le 1$ and $v_{n_1+1,2}\ge 1$;
    \item $m_1/n_1\ge m_2/n_2$ and $m_2>0$ and $n_2>0$;
    \item $m_1+m_2\ge n$.
\end{itemize}

Thoughout the rest of the paper, we always assume that the input satisfies the above constraints.
Our results might not hold true if these constraints are not satisfied.


\paragraph{Proper allocations.}

In our work, the Pareto-optimality is achieved in all of the constructed allocations, both intermediate ones (that might not be EFX) and the resulting one (that is necessary EFX).
All of our allocations fall under the notion given below.

\begin{definition}[Proper allocation]\label{def:proper-distribution}
An allocation of goods $\{(x_{i,1},x_{i,2})\}$ is \emph{proper}, if there exists $t\in[n]$
such that
\begin{itemize}
    \item for each $i\in[t-1]$, $x_{i,2}=0$, and
    \item for each $i\in [t+1, n]$, $x_{i,1}< v_{t,2}$.
\end{itemize}
\end{definition}

For example, the allocation in \Cref{obs:too-few-items} is proper (choose $t=m_1$).
As announced above, we will prove that all proper allocations are PO,
as formulated below.

\begin{theorem}\label{thm:proper-is-po}
    If allocation is proper, then it is Pareto-optimal.
\end{theorem}

\paragraph{Prioritized equitable allocations.}

The following notion is basic to our allocation constructions.
It formally defines the process of dividing the identical goods between specific agents in (almost) equal parts, where some specific agents should receive the greater part.

The notion below defines this process formally.
Note that we use it for \emph{partial} allocation of goods, while \emph{complete} allocations are formed via performing several partial allocations consecutively.

\begin{definition}[Prioritized equitable allocation, PEA]
    For integers $c\in [2], q\in [m_t], s \in [n]$, and a sequence of distinct agents $a_1, a_2,\ldots, a_s\in [n]$, we say that we give  $q$ goods of type $c$ \emph{equitably prioritized} to the agents $a_1, a_2, \ldots, a_s$, if
    \begin{itemize}
        \item for each $i\in [s]$, agent $a_i$ receives either $\lfloor q/s\rfloor$ or $\lceil q/s\rceil$ goods of type $c$;
        \item for each $1\le i < j \le s$, agent $a_j$ receives at least as many goods of type $c$ as $a_i$.
    \end{itemize}
\end{definition}

Equivalently, if we give $q$ goods of type $c$ to $s$ agents $a_1,a_2,\ldots, a_s$ equitably prioritized, and $r$ is the remainder of division of $q$ by $s$, then agents $a_1, a_2, \ldots, a_{s-r}$ receive $\lfloor q/s\rfloor$ goods each, and agents $a_{s-r+1},\ldots, a_{s}$ receive $\lceil q/s\rceil$ goods each.

\paragraph{Split allocations.}

We move on to the central allocation construction of our algorithm, the \emph{split allocations}.
Simply speaking, in a split allocation, only a specified agent $t$ (the ``split point'') can receive goods of both types simultaneously.
Agents $i<t$ or agents $i>t$ receive only goods of the first type or the second type respectively.

A split allocation is uniquely identified by two integers: $t\in[n]$ and $k\in [m_2]$ --- the number of goods of the second type given to agent $t$.
These integers should fit under specific constraints.
We define all valid pairs of integers $\mathcal{T}\subset [n]\times [m_2]$:
$$\mathcal{T}=\left\{(t,k): v_{t,2}\ge 1, t<n, k\le \frac{m_2}{n-t+1}\right\}\cup \{(n,m_2)\}.$$ 

Note that $\mathcal{T}$ is not empty since $n_2>0$ and $m_2>0$.
We now give formal definition to split allocations.

\begin{table}
\centering
{\small
\begin{tabular}{|c|c|c|}
    \hline
    \multicolumn{3}{|c|}{$(t,k)$-split-allocation as PEA sum} \\
    \hline
    \begin{tabular}[c]{@{}c@{}} Type\\$(c)$\end{tabular}  & 
    \begin{tabular}[c]{@{}c@{}} Quantity\\$(q)$\end{tabular}  &  
    \begin{tabular}{c}Agent seq. \\ $(a_1,\ldots, a_s)$\end{tabular}
     \\
    \hline
      $2$ & $k$ & $t$ \\
    \hline
     $2$& $m_2-k$ & $t+1,\ldots, n$ \\
    \hline
     $1$ & $\min\{m_1, \lceil kv_{t,2}\rceil\cdot  (t-1)\}$ & $1,\ldots, t-1$ \\
    \hline
    $1$ & $\max\{m_1 - \lceil kv_{t,2}\rceil \cdot (t-1)\}$ & $1,\ldots, t$ \\
    \hline
\end{tabular}}
\caption{Expessing $(t,k)$-split-allocation as a series of prioritized equitable allocations.}\label{fig:split-via-pea}
\end{table}

\begin{definition}[Split allocations]\label{def:split-distribution-not-numerical}
For $(t,k)\in\mathcal{T}$, the \emph{$(t,k)$-split-allocation} is defined by
\begin{enumerate}
    \item Give $k$ goods of the second type to agent $t$.
    \item Give $(m_2-k)$ remaining goods of the second type equitably prioritized to agents $t+1, t+2,\ldots, n$.
    \item Let $p=\lceil k\cdot v_{t,2} \rceil$.
    Give $\min\{p(t-1), m_1\}$ goods of the first type equitably prioritized to agents $1,2,\ldots, t-1$.
    \item Give remaining $\max\{0,m_1-p(t-1)\}$ goods of the first type equitably prioritized to agents $1,2,\ldots, t$.
\end{enumerate}
\end{definition}

We summarize that the $(t,k)$-split-allocation is expressed as a sequence of four prioritized equitable allocations in \Cref{fig:split-via-pea}.
Note that the agent $t$ can receive $0$ goods of the first type, as demonstrated in \Cref{fig:split-distrib-a}.

\begin{table}[ht]
    \small
    \centering
    \begin{tabular}{r|c|c|c|c|c|c|c|c|c|}
        & \multicolumn{1}{c}{$1$} & \multicolumn{1}{c}{$\ldots$}  & \multicolumn{1}{c}{$t-1$}  & \multicolumn{1}{c}{$t$} & \multicolumn{1}{c}{$t+1$}  & \multicolumn{1}{c}{$\ldots$}  & \multicolumn{1}{c}{$n$} \\
        \hline
         $1$ & $\lfloor \frac{m_1}{t-1}\rfloor$ & $\ldots$ & $\lceil \frac{m_1}{t-1}\rceil$ & $0$ & $0$ & $\ldots$ & $0$ \\
         \cline{2-8}
         $2$ & $0$&$\ldots$&$0$&$k$&$\lfloor \frac{m_2-k}{n-t}\rfloor$&$\ldots$&$\lceil \frac{m_2-k}{n-t}\rceil$\\
         \cline{2-8}
    \end{tabular}
    \caption{Structure of $(t,k)$-split-allocations with $m_1\le p(t-1)$. The number in $i$\textsuperscript{th} column and $j$\textsuperscript{th} row equals $x_{i,j}$, the number of goods of type $j$ given to agent $i$.}
    \label{fig:split-distrib-a}
\end{table}

The following simple observation explains why split allocations are PO.

\begin{observation}
    For each $(t,k)\in \mathcal{T}$, the $(t,k)$-split-allocation is proper.
\end{observation}

\paragraph{Envy direction.}
Our algorithm will encounter several split allocations.
If at least one of split allocations it encounters is EFX, the algorithm will be fine.

But if the $(t,k)$-split-allocation is not EFX for $(t,k)\in \mathcal{T}$, what can we say about the envy?
We will show that our design of split allocations is such that it is not possible that agent $i<t$ is envious and agent $j>t$ is envious simultaneously in the $(t,k)$-split-allocation.
We proceed to formal notions.

\begin{definition}[Left-envious (LE) and right-envious (RE) allocations]
We say that an allocation $X$ is \emph{left-envious}, or \emph{LE} (\emph{right-envious}, or \emph{RE}), if $X$ is not EFX, and there exists $i,j$ such that agent $j$ envies (up to any item) agent $i$ for $i < j$ ($i > j$).
\end{definition}

 We will prove that no $(t,k)$-split-allocation can have envy in both directions.

\begin{restatable}{theorem}{envyOneDirection}\label{thm:envy-one-direction}
    For each $(t,k)\in \mathcal{T}$, the $(t,k)$-split-allocation cannot be LE and RE simultaneously.
\end{restatable}

There are two natural ``extremal'' cases of split allocations.
One case is the $(n,m_2)$-split-allocation, where all goods of the second type are given to agent $n$.
This allocation maximizes the total utility provided by goods of the second type.
The other case is when goods of the second type receive the leftmost positions possible for a split allocation.
This case is the $(t,\lfloor m_2/(n-t+1)\rfloor)$-split-allocation, where $t\in [n]$ minimum possible such that $v_{t,2}\ge 1$ and $m_2\ge n-t+1$ (there can be no $(t',k')\in \mathcal{T}$ with $t'<t$).
Following this logic, we define orderings of split allocations.

\begin{definition}[Complete ordering of split allocations]
A complete linear ordering $\prec$ on the set $\mathcal{T}$ is defined by $$(t_1, k_1)\prec (t_2,k_2)$$ if and only if $t_1 < t_2$ or $t_1=t_2$ \emph{and} $k_1>k_2$.
\end{definition}

The two split allocations described above are naturally maximal and minimal elements in $\mathcal{T}$ with respect to $\prec$. 
Our second result on envy directions uncovers that these two extremal allocations have specific envy directions (if they are not EFX).

\begin{restatable}{theorem}{envyLeftRightBounds}\label{thm:left-and-right-bounds}
    Let $(t_L,k_L)$ be the minimum element in $\mathcal{T}$ with respect to $\prec$.
    Let $(t_R,k_R)$ be the maximum element in $\mathcal{T}$ with respect to $\prec$.
    Both of the following is true:
    \begin{enumerate}
        \item The $(t_L, k_L)$-split-allocation is either EFX or LE,
        \item The $(t_R, k_R)$-split-allocation is either EFX or RE.
    \end{enumerate}
\end{restatable}

\paragraph{Reallocation.}
To get intuition behind the final concept, assume that none of the split allocations is EFX.
Then \Cref{thm:left-and-right-bounds} guarantees that the $(t_L,k_L)$-split-allocation and the $(t_R,k_R)$-split-allocation are left-envious and right-envious respectively.
Since all split allocations are either LE or RE, there should be $(t,k)\in \mathcal{T}$ such that the $(t,k)$-split-allocation is LE, while for its immediate successor (w.r.t.\ $\prec$) $(t',k')\in \mathcal{T}$, the $(t',k')$-split-allocation is RE.
That is, $(t,k)$ and $(t',k')$ form a point where the ``envy direction changes''.

This can be, for example, $(t,k)\in \mathcal{T}$ for $t<n$ and $k\ge 2$, and $(t,k-1)\in \mathcal{T}$, such that the $(t,k)$-split-allocation is LE and the $(t,k-1)$-split-allocation is RE.

Our final allocation construction slightly transforms such $(t,k)$-split-allocations by redistributing some goods of the first type from agents to the left of $t$ to agents to the right of $t$.
We now define this construction formally (see also \Cref{fig:reallocation-pea}).

\begin{table}
\centering{\small
\begin{tabular}{|c|c|c|}
    \hline
    \multicolumn{3}{|c|}{$(t,k)$-reallocation as PEA sum} \\
    \hline
     \begin{tabular}[c]{@{}c@{}} Type\\$(c)$\end{tabular}  & 
    \begin{tabular}[c]{@{}c@{}} Quantity\\$(q)$\end{tabular}  &  
    \begin{tabular}{c}Agent seq. \\ $(a_1,\ldots, a_s)$\end{tabular}
     \\
    \hline
      $2$ & $k$ & $t$ \\
    \hline
     $2$& $m_2-k$ & $t+1,\ldots, n$ \\
    \hline
     $1$ & $\lceil dv_{t,2}\rceil \cdot (t-1)$ & $1,\ldots, t-1$ \\
    \hline
    $1$ & $\lceil dv_{t,2}\rceil-p$ & $t$ \\
    \hline
    \multirow{3}{*}{$1$} & \multirow{3}{*}{$m_1-\lceil dv_{t,2}\rceil \cdot t+p$} & $1,\ldots, \ell$ \\

    & & or  \\
    & & $1,\ldots,t-1,t+1,\ldots,\ell,t$ \\
    \hline
\end{tabular}}
\caption{Expressing $(t,k)$-reallocation as a series of prioritized equitable allocations.}\label{fig:reallocation-pea}
\end{table}

\begin{definition}[$(t,k)$-reallocation]
For $(t,k)\in \mathcal{T}$ such that $t<n$ and the $(t,k)$-split-allocation $X$ satisfies $$m_1\ge \lceil d\cdot v_{t,2} \rceil\cdot t -p,$$ where $p=\lceil kv_{t,2}\rceil$ and $d=x_{t+1,2}=\lfloor (m_2-k)/(n-t)\rfloor$, the \emph{$(t,k)$-reallocation} is obtained by 

\begin{enumerate}
    \item Start from $X$, but take away all items of the first type from each agent.
    \item Let $\ell\in [n]$ be maximum possible such that $x_{\ell,2}=d$.
    \item Give $\lceil d\cdot v_{t,2}\rceil$ items of the first type to each agent in $[t-1]$.
    \item Give $\lceil d\cdot v_{t,2}\rceil-p$ first-type items to agent $t$.
    \item Give the remaining $m_1-\lceil d\cdot v_{t,2}\rceil +p$ items of the first type equitably prioritized to:
    \begin{itemize}
        \item Agents $1,2,\ldots, \ell$, in case $\lceil d v_{t,2}\rceil - p>(d-k)\cdot v_{t,2}$;
        \item Agents $1,2,\ldots, t-1,t+1,\ldots,\ell,t$,
        otherwise.
    \end{itemize}
\end{enumerate}
\end{definition}

Note that the definition of $(t,k)$-reallocations forces additional constraints on $(t,k)$.
We will prove the following lemma, that guarantees that $(t,k)$ satisfies these (and other) constraints, if it corresponds to the point of ``envy direction change''.


\begin{restatable}{lemma}{envyDirectionChange}\label{lem:envy-direction-change}
    Let $(t,k),(t',k')\in\mathcal{T}$ be such that the $(t,k)$-split-allocation is left-envious, but $(t',k')$-split-allocation is right-envious.
    Moreover, there is no $(t'',k'')\in\mathcal{T}$ that
    satisfies $(t,k)\prec (t'',k'')\prec(t',k')$.
    Then $t<n$ and
    \begin{multline*}
    \lceil x_{t+1,2}\cdot v_{t,2}\rceil \cdot t-\lceil kv_{t,2}\rceil + 1\le m_1\le  \lceil x_{t+1,2}\cdot v_{t,2}\rceil \cdot t-\lceil (k-1) v_{t,2}\rceil.
    \end{multline*}
\end{restatable}

The final of our results is the following theorem.
If combined with \Cref{lem:envy-direction-change}, it grants an EFX+PO allocation.

\begin{restatable}{theorem}
{reallocationefxpo}\label{thm:reallocation-efx-po}
    Let $X$ be the $(t,k)$-reallocation for some $(t,k)\in \mathcal{T}$.
    If $$m_1 \le \lceil x_{t+1,2}\cdot v_{t,2}\rceil \cdot t-\lceil (k-1)\cdot v_{t,2}\rceil,$$
    then $X$ is EFX+PO.
\end{restatable}

As discussed before in this section, the Pareto-optimality will follow from the fact that any $(t,k)$-reallocation is proper.
We will prove this fact as a part of the proof of \Cref{thm:reallocation-efx-po}.

\section{The Algorithm}

We summarize the constructions and properties of the previous section into an algorithm that constructs the desired allocation.
It proves, in particular, that EFX+PO allocations  always exist for two types of goods, given that all agents' utilities are positive.

\begin{algorithm}[h]
	\SetKwIF{IfOr}{OrIf}{ElseOr}{if}{or}{or}{}{}
	$t_L\gets$ smallest integer in $[n]$ with $(t_L,*)\in \mathcal{T}$\label{line:tl}\; 
    $X_L\gets $ the $(t_L,\lfloor \frac{m_2}{n-t_L+1} \rfloor)$-split-allocation\label{line:xl}\;
    \lIf{$X_L$ is EFX}{\Return{$X_L$}}
    $t_R\gets n$\;
    $X_R\gets $ the $(t_R,m_2)$-split-allocation\label{line:xrcheck}\;
    \lIf{$X_R$ is EFX}{\Return{$X_R$}}
    \tcc{binary search to find $t$}
    \While{$t_R-t_L>1$}{
        $t_M\gets \lfloor (t_L+t_R)/2\rfloor$\label{line:beg-bs-t}\;
        $X_M\gets$ the $(t_M,\lfloor\frac{m_2}{n-t_M+1}\rfloor)$-split-allocation\;
        \lIf{$X_M$ is EFX}{\Return{$X_M$}}
        \leIf{$X_M$ is LE}{$t_L\gets t_M$}{$t_R\gets t_M$}\label{line:end-bs-t}
    }
    $t\gets t_L$ and $k_L\gets \lfloor\frac{m_2}{n-t+1}\rfloor$ and $k_R\gets 1$\label{line:after-first-bs}\;
    \eIf{$(t,k_R)$-split-allocation is RE}{
     \tcc{binary search to find $k$}
        \While{$k_L-k_R>1$\label{line:beg-bs-k}}{
        $k_M\gets \lfloor (k_L+k_R)/2\rfloor$\;
        $X_M\gets$ the $(t,k_M)$-split-allocation\;
        \lIf{$X_M$ is EFX}{\Return{$X_M$}}
        \leIf{$X_M$ is LE}{$k_L\gets k_M$}{$k_R\gets k_M$}
\label{line:end-bs-k}
    }
            $k\gets k_L$\;
    }{

    $X_R\gets $ the $(t,k_R)$-split-allocation\label{line:some}\;
    \lIf{$X_R$ is EFX}{\Return{$X_R$}}     \tcc{LE changes to RE between $(t,k_R)$ and $(t+1,\lfloor m_2/(n-t)\rfloor)$}   
        $k\gets k_R$\;\label{line:other}
    }

    \Return{the $(t,k)$-reallocation.\label{line:final-line}}
	\caption{The algorithm constructing EFX+PO allocation on preprocessed inputs.}
	\label{alg:efxpo}
\end{algorithm}

\begin{theorem}
Given preprocessed input, an EFX+PO allocation can be found in $\mathcal{O}(\log n + \log  m)$ time.
\end{theorem}
\begin{proof}
We present an algorithm that is based on envy directions of split allocations.
Its pseudocode can be found in \Cref{alg:efxpo}, and we will refer to its lines throughout the proof.

We have to clarify first why the algorithm is able to work in logarithmic time, given that even just outputting an allocation ($2n$ integers) takes linear time.
In the following claim, we explain why this is possible; moreover, we can even determine whether a split allocation is EFX in constant time.

    \begin{claim}
        Given $(t,k)\in \mathcal{T}$, we can determine in $\mathcal{O}(1)$ time whether the $(t,k)$-split-allocation is EFX, left-envious or right-envious. 
    \end{claim}
    \begin{claimproof}
        By definition of split allocations, there are $b\le 5$ distinct bundles that agents receive in the $(t,k)$-split-allocation $X$.
        Additionally, there are $b$ contiguous segments of agent numbers $[a_1, a_2-1], [a_2, a_{3}-1], \ldots, [a_b, n]$, where $a_1=1$ and $a_1< a_2 <\ldots < a_b$.
        Agents within the same segment receive exactly the same bundle.
        Note that $b$ and $a_1, a_2,\ldots, a_b$ are easily computable in $\mathcal{O}(1)$ time from $(t,k)$ following \Cref{def:split-distribution-not-numerical}, as well as the bundles themselves.
        
        Moreover, if there is envy (up to any item) from agent $i$ to agent $j$, and
        agents $i-1$, $i$, $i+1$ receive exactly the same bundle, then there is necessary envy (up to any item) from agent $i-1$ to agent $j$ or from agent $i+1$ to agent $j$.

        To see the last paragraph, assume agent $i$ envies agent $j$ up to any item.
        Equivalently, there is a choice of integers $z_1,z_2\in\{0,1\}$ such that $z_1+z_2=1$ and $z_c\le x_{j,c}$, and
        $$x_{i,1}+x_{i,2}\cdot v_{i,2}<(x_{j,1}-z_1)+(x_{j,2}-z_2)\cdot v_{i,2},$$
        or, equivalently,
        \begin{equation}\label{eq:consecutive-envy}
        x_{i,1}-x_{j,1}+z_1 < (x_{j,2}-x_{i,2}-z_2)\cdot v_{i,2}.
        \end{equation}
        Since $v_{i-1,2}\le v_{i,2} \le v_{i+1,2}$ and $x_{i-1,c}=x_{i,c}=x_{i+1,c}$ for each $c \in [2]$, we have that \eqref{eq:consecutive-envy} holds if we replace $i$ with either $i+1$ (if the right part is non-negative), or $i-1$ (if the right part is negative).
        Equivalently, either agent $i+1$ or agent $i-1$ envies agent $j$ up to any item.

        It follows that there is no need to check whether an agent $i$ with $a_s<i < a_{s+1}-1$ envies (up to any item) any other agent.
        The only agents we have to check are the agents with numbers in $\{a_1, a_   2-1, a_2, \ldots, a_{b-1}, a_b, n\}.$
        There are $b$ distinct bundles in $X$, so for each agent we have to make $b-1$ comparisons (whether an agent $i$ envies an agent with this specific bundle).
        
        In total, one should make only a total of $2b(b-1)$ comparisons between agents and bundles.
        Each comparison is done in constant time.
    \end{claimproof}

    We move on to the description of the algorithm itself.
    The algorithm first evaluates $t_L$ from the statement of \Cref{thm:left-and-right-bounds} (Line \ref{line:tl}).
    Since $$t_L=\max\{\min\{t \in[n]:  v_{t,2}\ge 1\}, n-m_2+1\},$$ $t_L$ is evaluated in $\mathcal{O}(\log n)$ time via lower bound binary search over $(v_{1,2},v_{2,2},\ldots, v_{n,2})$.

    The other integers from the statement of \Cref{thm:left-and-right-bounds} are computable in $\mathcal{O}(1)$ time, and the algorithm checks whether the $(t_L,k_L)$- or the $(t_R,k_R)$-split-allocation is EFX (Lines \ref{line:xl}-\ref{line:xrcheck}).
    If the algorithm does not return here, from \Cref{thm:left-and-right-bounds} we know that the allocations are LE and RE respectively.

    The algorithm then aims to find $(t,k)\in \mathcal{T}$ satisfying \Cref{lem:envy-direction-change}.
    This is done via two binary searches.
    In the first binary search, the algorithm finds $t$, and in the second one --- the algorithm finds $k$ (given that $t$ is known).
    During these searches, the algorithm might encounter a split allocation that is EFX.
    If this happens, the algorithm returns such allocation as its final answer.

    We move on to discussion of the first binary search.
    To perform this search, the algorithm treats $t_L$ and $t_R$ as binary search bounds and modify them correspondingly, until it reaches $t_R-t_L=1$.
    It keeps the following invariant: the $(t_L,\lfloor m_2/(n-t_L+1)\rfloor)$-split-allocation is LE, and  the $(t_R,\lfloor m_2/(n-t_R+1)\rfloor)$-split-allocation is RE.
    In a single iteration of the binary search, the algorithm takes $t_M$ in the middle between $t_L$ and $t_R$, checks whether the $(t_M,\lfloor m_2/(n-t_M+1)\rfloor)$-split-allocation is EFX, LE, or RE, and returns the correct allocation, puts $t_L$ equal to $t_M$, or puts $t_R$ equal to $t_M$ correspondingly (Lines \ref{line:beg-bs-t}-\ref{line:end-bs-t}).

    After the binary search, the algorithm puts $t:=t_L$.
    We know that the desired ``envy direction change'' point has form $(t,k)$ for some integer $k$ between $\lfloor m_2/(n-t+1)\rfloor$ and $1$ (because the successor of $(t,1)$ w.r.t.\ to $\prec$ is RE).
    If $(t,1)$ is right-envious allocation, then the algorithm performs binary search over $k$ starting with $k_L=\lfloor m_2/(n-t+1)\rfloor$ and $k_R=1$ in exactly the same way as in the first binary search (Line \ref{line:beg-bs-k}-\ref{line:end-bs-k}).  An only little difference from the first binary search is that $k_L>k_R$ (instead of $k_L<k_R$), as imposed by definition of $\prec$.


    If the $(t,1)$-split-allocation is not left-envious, then its either EFX (and forms the correct solution), or $k=1$ is viable choice for \Cref{lem:envy-direction-change} (Lines \ref{line:some}-\ref{line:other}).

    In either of the two cases, the algorithm finds $(t,k)$ viable for \Cref{lem:envy-direction-change} (or encounters an EFX split allocation and returns before).
    The algorithm constructs the $(t,k)$-reallocation and returns it as a final solution (Line \ref{line:final-line}).
    This allocation is EFX+PO by \Cref{thm:reallocation-efx-po}.

    Similarly to split allocations, the $(t,k)$-reallocation is found in $\mathcal{O}(1)$ time, since there is at most a constant number of bundles given to contiguous segments of agents.
\end{proof}

\section{Pareto Optimality of Proper Allocation}

In this section we prove \Cref{thm:proper-is-po}, stating that any proper allocation (\Cref{def:proper-distribution}) is PO.

We start with the following simple lemma that allows natural restrictions on a Pareto-dominating allocation.

\begin{lemmaO}\label{lem:reasonable-po}
    Let $X=\{(x_{i,1}, x_{i,2})\}$ be an item allocation.
    If $X$ is not PO, then there exists a Pareto-dominating item allocation $Y=\{(y_{i,1},y_{i,2})\}$, such that for each $i\in[n], j\in [2]$

    $$\Delta_{i,1}\cdot \Delta_{i,2}\le 0,$$
    where $\Delta_{i,j}=y_{i,j}-x_{i,j}$.
\end{lemmaO}
\begin{proofO}
    The proof is by contradiction.
    Let $X$ be a given (non-Pareto-optimal) allocation and assume that no allocation $Y$ satisfies the lemma statement.
    Then let $Y$ be a allocation that Pareto-dominates $X$, with minimum possible $$\sum_{i\in[n],j\in[2]} |\Delta_{i,j}|.$$

    By our assumption, there exists $i^*\in [n]$ such that $\Delta_{i^*,1}\cdot \Delta_{i^*,2}>0.$
    If both $\Delta_{i^*,1},\Delta_{i^*,2}<0$, then utility of agent $i^*$ in $X$ is at least $\sum_{j\in[2]}|\Delta_{i^*,j}|\cdot v_{i^*,j}>0$ greater than its utility in $Y$, and this is impossible.

    Consequently, we have $\Delta_{i^*,1},\Delta_{i^*,2}>0$.
    Since $\sum_{i\in[n]} \Delta_{i,1}=0$, there exists $k^*\in [n]$ such that $\Delta_{k^*,1}<0$
    Now construct another allocation $Y'$ by transferring exactly one item of type $1$ from agent $i^*$ to agent $k^*$ in $Y$.

    Formally, we define $\Delta'_{i^*,1}=\Delta_{i^*,1}-1\ge 0$ and $\Delta'_{k^*,1}=\Delta_{k^*,1}+1$.
    For each other $i\in [n], j\in[2]$, we put $\Delta'_{i,j}=\Delta_{i,j}$.
    Then $Y'=\{(y'_{i,1},y'_{i,2})\}$ is given by $y'_{i,j}=x_{i,j}+\Delta'_{i,j}$ similarly to $Y$ and $\Delta$.
    Note that $Y'$ is valid since $y'_{i,j}\ge 0$ for every $i\in[n],j\in[2]$, and $\sum_{i\in[n]}\Delta'_{i,j}=0$ for each $j\in[2]$.

    We claim that $Y'$ Pareto-dominates $X$.
    To see this, note that the only agent that receives less utility in $Y'$ than in $Y$, is agent $i^*$.
    Then we only have to verify that $i^*$ still receives more utility in $Y'$ than in $X$.
    Indeed, the difference between these utilities is at least
    $$\sum_{j\in[2]} \Delta'_{i^*,j}\cdot v_{i^*,j}\ge \Delta'_{i^*,2}\cdot v_{i^*,2}=\Delta_{i^*,2}\cdot v_{i^*,2}>0.$$

    Therefore, $Y'$ Pareto-dominates $X$, but $\sum
    |\Delta'_{i,j}|=\sum|\Delta_{i,j}|-2.$
    This contradicts the initial choice of $Y$.
\end{proofO}

Armed up with \Cref{lem:reasonable-po}, we move on to proving that any \emph{proper} allocation is PO.

\begin{proof}[Proof of \Cref{thm:proper-is-po}]
    Let $X$ be a proper allocation, where agent $i$ receives $x_{i,1},x_{i,2}$ goods of types $1$ and $2$ respectively, for each $i\in [n]$.
    Let $t\in [n]$ be the integer guaranteed by \Cref{def:proper-distribution}: $t$ is such that for each $i\in [t-1]$ $x_{i,2}=0$, and, for each $i\in [t+1, n]$, it holds $x_{i,1}<v_{t,2}$.

    Targeting towards a contradiction, assume that $X$ is not PO, and there exists a allocation $Y=\{(y_{i,1},y_{i,2})\}$ that Pareto-dominates $X$ and satisfies \Cref{lem:reasonable-po}.
    We use $\Delta_{i,j}=y_{i,j}-x_{i,j}$.
    We have $\Delta_{i,1}\cdot \Delta_{i,2}\le 0$ for each $i\in [n]$.

    For each $j\in [2]$, define $I_j=\{i: \Delta_{i,j}>0\}$, that is, $I_j$ is a set of agents that receive extra items of type $j$ in $Y$ (when compared to $X$).
    Before showing a contradiction, we prove two useful claims on properties of $\Delta_{i,j}$.
    Both claims together show that the ratio between the number of exchanged goods of two types is tied to $v_{t,2}$.
\begin{claim}\label{claim:given-one-lb-two}
It holds
    $$\sum_{i\in I_1} \Delta_{i,1}\ge v_{t,2}\cdot \sum_{i\in I_1}(-\Delta_{i,2}).$$
     If utility of some agent $i\in I_1$ is strictly greater in $Y$ than in $X$, then the inequality is strict.
\end{claim}

\begin{claimproof}
   Consider arbitrary agent $i\in I_1$.
   Its utility in $Y$ is not less than its utility in $X$, hence $\Delta_{i,1}\cdot v_{i,1}+\Delta_{i,2}\cdot v_{i,2}\ge 0$, equivalently $\Delta_{i,1}+\Delta_{i,2}\cdot v_{i,2}\ge 0$.
   If its utility is strictly greater in $Y$, then this inequality is strict.

   If $\Delta_{i,2}=0$, then $\Delta_{i,1}\ge 0$, and $\Delta_{i,1}\ge -\Delta_{i,2}\cdot v_{t,2}$ holds.
   If $\Delta_{i,2}<0$, then $x_{i,2}>0$, and by definition of $t$ we have $i\ge t$.
   Consequently, $v_{i,2}\ge v_{t,2}$.
   Then obtain
   $$\Delta_{i,1}\ge -\Delta_{i,2}\cdot v_{i,2}\ge -\Delta_{i,2}\cdot v_{t,2}.$$
   The first part of the claim follows by summing over all $i\in I_1$.

   To see the second part of the claim, note that $\Delta_{i,1}>-\Delta_{i,2}\cdot v_{t,2}$ is strict for the choice of $i\in I_1$ with greater utility in $Y$.
   The resulting sum inequality also becomes strict.
   The proof is complete.
\end{claimproof}

The second claim gives a symmetrical lower bound via different arguments.

\begin{claimO}\label{claim:given-two-lb-one}
It holds
    $$\sum_{i\in I_2}\Delta_{i,2}\ge \frac{1}{v_{t,2}}\cdot \sum_{i\in I_2}(-\Delta_{i,1}).$$
    If utility of some agent $i\in I_2$ is strictly greater in $Y$ than in $X$, then the inequality is strict.
\end{claimO}

\begin{claimproofO}
    Consider an agent $i\in I_2$.
    Similarly to the proof of the previous claim, it is enough to show that $\Delta_{i,2}\ge -\Delta_{i,1}/v_{t,2}$.

    We consider two cases depending on whether $i\le t$.
    If $i\le t$, then $v_{i,2}\le v_{t,2}$.
    While utility of agent $i$ in $Y$ is not less than in $X$, we have $\sum_{j\in [2]}\Delta_{i,j}\cdot v_{i,j}\ge 0$, hence $$\Delta_{i,2}\cdot v_{t,2}\ge\Delta_{i,2}\cdot v_{i,2}\ge  -\Delta_{i,1}\cdot v_{i,1}=-\Delta_{i,1}$$ as desired.

    The remaining case is $i>t$.
    By definition of $t$, $x_{i,1}<v_{t,2}$.
    Since $\Delta_{i,2}>0$, it holds $\Delta_{i,1}\le 0$, and $0\le y_{i,1}\le x_{i,1}<v_{t,2}.$
    Then $-\Delta_{i,1}=x_{i,1}-y_{i,1}< v_{t,2}.$
    Finally, we know that $\Delta_{i,2}\ge 1$, and $\Delta_{i,2}\ge -\Delta_{i,1}/v_{t,2}$ follows.
    This finishes the proof of the first part of the claim.

    The second part is the same as in the proof of \Cref{claim:given-one-lb-two}.
    If for at least one $i\in I_2$ we have strictly greater utility, this choice of $i$ gives strict $\Delta_{i,2}>-\Delta_{i,1}/ v_{t,2}$, and the resulting sum is strict.
\end{claimproofO}

To combine \Cref{claim:given-one-lb-two} and \Cref{claim:given-two-lb-one}, note that $\sum_{i\in[n]}\Delta_{i,j}=0$ for each $j\in [2]$ by definition of $\Delta_{i,j}$.
By definition of $I_j$ and since $Y$ Pareto-dominates $X$, for each $i\in [n]\setminus (I_1\cup I_2)$ we have $\Delta_{i,1}=\Delta_{i,2}=0$.
Consequently, for each $j\in [2]$ we have
$$\sum_{i\in [n]}\Delta_{i,j}=\sum_{i\in I_1}\Delta_{i,j}+\sum_{i\in I_2}\Delta_{i,j},$$
and
$$-\sum_{i\in I_j} \Delta_{i,j}=\sum_{i\in I_{3-j}} \Delta_{i,j}.$$
Substituting this in the right parts of \Cref{claim:given-one-lb-two} (for $j=1$) and \Cref{claim:given-two-lb-one} (for $j=2$), we obtain
$$\sum_{i\in I_1}\Delta_{i,1}=v_{t,2}\cdot \sum_{i\in I_2}\Delta_{i,2},$$
and neither \Cref{claim:given-one-lb-two} nor \Cref{claim:given-two-lb-one} can give strict inequality.
This means that there is no agent $i\in I_1\cup I_2$ can have greater utility in $Y$ when compared to its utility in $X$.

But $Y$ Pareto-dominates $X$, so there is some agent $i^*\in [n]\setminus (I_1\cup I_2)$ whose utility is strictly greater.
But this agent has both $\Delta_{i^*,1}\le 0$ and $\Delta_{i^*,2}\le 0$, that is, he did not receive any new item in $Y$ when compared to $X$, and his utility in $Y$ cannot exceed his utility in $X$.
The obtained contradiction finishes the proof of the theorem.
\end{proof}

\section{Envy Directions in Split Allocations}

This section is devoted to properties of envy (up to any item) and envy directions in split allocations.
Our main goal is to prove \Cref{thm:envy-one-direction} and \Cref{thm:left-and-right-bounds}, that are fundamental to our algorithm.

Before proceeding, we fomulate split allocations numerically for referencing throughout the proof.

\begin{definition}[Numerical definition of split allocations]\label{def:split-distribution}
For $(t,k)\in\mathcal{T}$, the $(t,k)$-split-allocation $\{(x_{i,1},x_{i,2})\}$ is given by
$$
x_{i,2}=
\begin{lcases}
    0, & \; \text{if } i < t,\\
    k, & \; \text{if } i = t,\\
    q_2, & \; \text{if } t < i < {n-r_2},\\
    q_2+1, & \; \text{if } i > n-r_2,
\end{lcases}
$$
where $m_2-k=q_2(n-t)+r_2$ for $0 \le r_2 < n-t$,
and, for $p=\lceil k\cdot v_{t,2}\rceil$,
\begin{enumerate}
\item[\bf \large(a)]
    in case of $m_1< p(t-1)$,
    $$
    x_{i,1}=
    \begin{lcases}
        0, & \; \text{if } i \ge t,\\
        q_1, & \; \text{if } i \le t-1-r_1,\\
        q_1+1, & \; \text{if } i \in [t-r_1, t-1],
    \end{lcases}
    $$
    where $m_1=q_1\cdot (t-1)+r_1$ for $0 \le r_1 < (t-1)$,
\item[\bf \large (b)] 
    in case of $m_1\ge p(t-1)$,
    $$
    x_{i,1}=
    \begin{lcases}
        0, & \; \text{if } i > t,\\
        q_1, & \; \text{if } i = t \text{ and } r_1=0,\\    
        q_1+1, & \; \text{if } i = t \text{ and } r_1>0,\\
        p+q_1, & \; \text{if } i \le \min\{t-r_1,t-1\},\\
        p+q_1+1, & \; \text{if } i \in [t-r_1+1,t-1],
    \end{lcases}
    $$
    where $m_1-p(t-1)=q_1\cdot t+r_1$ for $0 \le r_1 < t$. 
\end{enumerate}
\end{definition}

We start with a series of lemmas.
The first simple lemma demonstrates that there can be no envy between agents that go before $t$ or between two agents that go after $t$ in a $(t,k)$-split-allocation.

\begin{lemmaO}\label{lem:no-envy-left-or-right}
    Let $X$ be a $(t,k)$-split-allocation for $(t,k)\in \mathcal{T}$.
    If there is envy (up to any item) between agents $i<j$ in $X$ for $i,j\in [n]$, then $i\le t$ and $j\ge t$.
\end{lemmaO}
\begin{proofO}
    The proof is by contradiction.
    Suppose that $X$ is not EFX and there is envy (up to one item) between agent $i$ and agent $j$ with $i<j$ in $X$, but $i>t$ or $j<t$ holds.

    Consider the case when $j<t$.
    By definition of split allocations, $|x_{i,1}-x_{j,1}|\le 1$ for each $i,j \in [t-1]$, while $x_{i,2}=x_{j,2}=0$.
    Hence, there cannot be envy (up to one item) between agent $i$ and agent $j$ since their bundles differ by at most one item.

    Then it should be the case when $i>t$.
    Similarly, in the $(t,k)$-allocation $X$ for any $i,j\in [t+1, n]$ we have $|x_{i,2}-x_{j,2}|\le 1$ and $x_{i,1}=x_{j,1}=0$.
    There cannot be envy between agent $i$ and agent $j$ if $i,j>t$.
    The proof is complete.
\end{proofO}

The next lemma restricts any left-envious (LE) split allocation to the case (b) from \Cref{def:split-distribution} of split allocations.

\begin{lemmaO}\label{lem:le-case-b}
    Let $X$ be the $(t,k)$-split-allocation for $(t,k)\in \mathcal{T}$.
    If $X$ is LE, then $m_1\ge p\cdot (t-1)$, where $p=\lceil k\cdot v_{t,2}\rceil$.
\end{lemmaO}
\begin{proofO}
    Targeting towards a contradiction, assume that $m_1<p(t-1)$, and $X$ is left-envous.
    That is, there is envy (up to any item) between agent $i$ and agent $j$ for $i<j$ in $X$, and agent $j$ envies agent $i$.
    By \Cref{lem:no-envy-left-or-right}, $i\le t$ and $j\ge t$.

    There are three possible cases.
    The first case is when $i<t$ and $j=t$.
    By definition of split allocations, we have $x_{i,1}\le \lfloor m_1/(t-1)\rfloor+1$ and $x_{i,2}=0$, hence $x_{i,1}\le p$.
    Agent $t$ received exactly $k$ items of the second type, so $x_{t,2}= k$.
    Then $p=\lceil x_{t,2} \cdot v_{t,2} \rceil$, so utility of $t$ in $X$ is at least $x_{t,2}v_{t,2}> p-1$.
    If we remove one arbitrary item (it's necessary a first-type item) from the bundle of $i$, the utility of the resulting bundle (in terms of agent $t$) is exactly $(x_{i,1}-1)\le p-1$.
    Hence, agent $t$ cannot envy agent $i<t$.
 
    The next case is when $i<t$ and $j> t$.
    By definition of the $(t,k)$-split-allocation, agent $j$ receives at least $\lfloor \frac{m_2-k}{n-t}\rfloor$ items of the second type, so 
    \begin{eqnarray*}
        x_{j,2}&&\ge \left\lfloor \frac{m_2-k}{n-t}\right\rfloor
        \\&&\ge\left\lfloor\frac{m_2-m_2/(n-t+1)}{n-t}\right\rfloor\\&&=\left\lfloor\frac{m_2}{n-t+1}\right\rfloor
        \ge k=x_{t,2},
    \end{eqnarray*}
 where the upper-bound on $k$ follows from $(t,k)\in\mathcal{T}$.
    That is, agent $j$ receives at least as many items of the second type as agent $t$.
    Then $j$ cannot envy $i$ similarly to the previous case.

    The remaining case is when $i=t$ and $j>t$.
    By definition of $X$, $x_{t,1}=0$.
    On the other hand, $x_{j,2}\ge x_{t,2}$ as described in the previous case.
    Hence, agent $j$ has as many items of each type as agent $t$ has.
    Agent $j$ cannot envy agent $t$.

    The list of the cases is exhausted.
    The obtained contradiction concludes the proof.
\end{proofO}

The final lemma in the series explains that $m_1\ge p(t-1)$ also rules out envies between any agent $i<t$ and agent $t$.

\begin{lemmaO}\label{lem:envy-case-b-right}
    Let $X$ be the $(t,k)$-split-allocation for $(t,k)\in \mathcal{T}$, and let $p=\lceil k\cdot v_{t,2}\rceil$.
    If $m_1\ge p\cdot (t-1)$, then for each $i\in [t-1]$ there is no envy (up to any item) between agent $i$ and agent $t$ in $X$.
\end{lemmaO}
\begin{proofO}
    We consider two cases of envy.
    In each of the cases, we come to a contradiction.

    \medskip\noindent\textbf{Agent $i$ envies agent $t$.} The first case is when agent $i$ envies (up to any item) agent $t$.
    If $x_{t,1}=0$, then the envy from $i$ to $t$ is equivalent to 
    \begin{equation}\label{eq:envy-xit-zero}
    x_{i,1}+(x_{i,2}+1)\cdot v_{i,2}<x_{t,2}\cdot v_{i,2}.
    \end{equation}
    From the definition of split allocations, we know that $x_{i,1}\ge p$.
    From definition of $p$, we know that $x_{t,2}\cdot v_{i,2}\le x_{t,2}\cdot v_{t,2}\le p$.
    Then from \eqref{eq:envy-xit-zero} we have $p<p$, which is a contradiction.
    Then it is necessary that $x_{t,1}>0$.

    Then the envy from $i$ to $t$ is equivalent to
    \begin{equation}\label{eq:envy-xit-nonzero}
    x_{i,1}+x_{i,2}\cdot v_{i,2}+\min\{1,v_{i,2}\}<x_{t,1}+x_{t,2}\cdot v_{i,2}.
    \end{equation}
    From definition of the $(t,k)$-split-allocation (\Cref{def:split-distribution} case (b)), we have that $x_{i,1}\ge p+q_1$ and $x_{i,2}=0$, while $x_{t,1}\le q_1+1$.
    Using these with \eqref{eq:envy-xit-nonzero}, obtain
    $$(p+q_1)+0\cdot v_{i,2}+\min\{1,v_{i,2}\}< (q_1+1)+x_{t,2}\cdot v_{i,2},$$
    and rewrite as
    \begin{equation}\label{eq:bound-on-p}
    p+\min\{1,v_{i,2}\}<x_{t,2}\cdot v_{i,2}+1.
    \end{equation}
    Use definition of $p$ and obtain 
    $$p+\min\{1,v_{i,2}\}<x_{t,2}\cdot v_{i,2}+1\le x_{t,2}\cdot v_{t,2} +1\le p+1.$$
    Consequently, $v_{i,2}\ge 1$ is not possible so we have $v_{i,2}<1$ and by \eqref{eq:bound-on-p} we have
    $$p+v_{i,2}<x_{t,2}\cdot v_{i,2}+1,$$
    so
    $$p<v_{i,2}\cdot (x_{t,2}-1)+1<(x_{t,2}-1)+1=x_{t,2}.$$
    From definition of $p$ we obtain $x_{t,2}\cdot v_{t,2}< x_{t,2}$, and this is a contradiction with $v_{t,2}\ge 1$ that comes from $(t,k)\in \mathcal{T}$.
    
    The case when agent $i$ envies agent $t$ up to any item is not possible.
    We move on to the remaining case.

    \medskip\noindent\textbf{Agent $t$ envies agent $i$.}
    In this case, the envy from agent $t$ to agent $i$ in $X$ is equivalent to
    \begin{equation}\label{eq:t-envies-i}
    (x_{t,1}+1)+x_{t,2}\cdot v_{t,2}<x_{i,1},
    \end{equation}
    since agent $i$ has no items of the second type.
    From definition of the $(t,k)$-split-allocation $X$, we have that $x_{i,1}-x_{t,1}\le p$ (we can get $p-1$ or $p$ depending on $i$ and $r_1$).
    Using this with \eqref{eq:t-envies-i}, obtain
    $$x_{t,2}\cdot v_{t,2}+1<p.$$
    This contradicts the definition of $p$, as $p=\lceil x_{t,2}\cdot v_{t,2}\rceil\le x_{t,2}\cdot v_{t,2}+1$.
    
    The proof of the case and the lemma is complete.
\end{proofO}

We are ready to prove our main result on envy directions in split allocations, \Cref{thm:envy-one-direction}.
For convenience of the reader, we restate it here first.

\envyOneDirection*

\begin{proof}
    Targeting towards a contradiction, assume that there exists $(t,k)\in\mathcal{T}$ such that the $(t,k)$-split-allocation $X=\{(x_{i,1},x_{i,2})\}$ is LE and RE simultaneously.

    By \Cref{lem:le-case-b}, we know that $m_1\ge p\cdot (t-1)$ holds for $p=\lceil x_{t,2}\cdot v_{t,2}\rceil$.
    In particular, $X$ comes from case (b) of \Cref{def:split-distribution}.
    Then, by \Cref{lem:envy-case-b-right} we have that there is no envy (up to any item) between any agent $i\in[t-1]$ and agent $t$ in $X$.
    Consequently, if there is an envy between agent $i$ and agent $j$ in $X$, where $i<j$, then $i\le t$ and $j>t$ necessarily holds.

    Under the initial assumption that $X$ is LE and RE simultaneously, we prove two contraversary claims.
    Recall that $q_1$ and $r_1$ from \Cref{def:split-distribution} of split allocations are non-negative integers satisfying $m_1=q_1\cdot (t-1)+r_1$.
    Similarly, $q_2$ and $r_2$ are two positive integers that satisfy $m_2-k=q_2\cdot(n-t)+r_2$.
    
    The first claim comes from $X$ being right-envious.
    \begin{claim}\label{claim:from-re}
    Both of the following is true:
    \begin{itemize}
        \item If $r_1>0$, then $p+q_1<q_2\cdot v_{t,2}$.
        \item If $r_1=0$, then $v_{t,2}\cdot x_{t,2}+q_1<q_2\cdot v_{t,2}$.
    \end{itemize}
    \end{claim}
    \begin{claimproof}
        Since $X$ is right-envious, there is an agent $i\in [t]$ that envies (up to any item) some agent $j\in[t+1,n]$.
        There are two cases depending on whether $i<t$ or $i=t$.
        
        If $i<t$, then agent $i$ has no items of the second type, and agent $j$ has no items of the first type.
        The envy from agent $i$ to agent $j$ is equivalent to $x_{i,1}+1\cdot v_{i,2}<x_{j,2}\cdot v_{i,2}$.
        From \Cref{def:split-distribution} we have that $x_{i,1}\ge p+q_1$ and $x_{j,2}\le q_2+1$.
        Consequently, $p+q_1+v_{i,2}<(q_2+1)\cdot v_{i,2}.$
        It follows $p+q_1<q_2\cdot v_{i,2}\le q_2\cdot v_{t,2}$, since $v_{i,2}\le v_{t,2}.$

        We have that if $i<t$, then $p+q_1<q_2\cdot v_{t,2}$.
        Since $p\ge k\cdot v_{t,2}=x_{t,2}\cdot v_{t,2}$, it also follows $v_{t,2}\cdot x_{t,2}+q_1 <q_2\cdot v_{t,2}$.
        Therefore, both parts of the claim follow if $i<t$.

        We now consider the case when agent $t$ envies (up to any item) agent $j>t$.
        The envy is equivalent to
        $$x_{t,1}+(x_{t,2}+1)\cdot v_{t,2}<x_{j,2}\cdot v_{t,2}.$$
        Recall that $x_{j,2}\le q_2+1$ and rewrite the above as
        \begin{equation}\label{eq:claim-left-envy}
            x_{t,1}+x_{t,2}\cdot v_{t,2}<q_2\cdot v_{t,2}.
        \end{equation}
               
        If $r_1=0$, then $x_{t,1}=q_1$, and \eqref{eq:claim-left-envy} becomes
        $$x_{t,2}\cdot v_{t,2} + q_1\le q_2\cdot v_{t,2},$$
        as required by the first part of the claim statement.

        If $r_1>0$, then $x_{t,1}=q_1+1$.
        Then, by $p< x_{t,2}\cdot v_{t,2}+1$ and \eqref{eq:claim-left-envy} obtain
        $$p+q_1<x_{t,2}\cdot v_{t,2}+1+q_1\le x_{t,1}+x_{t,2}\cdot v_{t,2}< q_2\cdot v_{t,2}.$$
        This proves the second part of the claim statement.
        This concludes the proof of the claim.
    \end{claimproof}

    The second claim comes from $X$ being left-envious.

    \begin{claim}\label{claim:from-le}
    Both of the following is true:
    \begin{itemize}
        \item If $r_1>0$, then $p+q_1>q_2\cdot v_{t,2}$.
        \item If $r_1=0$, then $v_{t,2}\cdot x_{t,2}+q_1>q_2\cdot v_{t,2}$.
    \end{itemize}
    \end{claim}
    \begin{claimproof}
        Since $X$ is left-envious, there is an agent $j>t$ that envies some agent $i\in[t]$.
        There are two cases depending on whether $i<t$ or $i=t$.

        Consider first the case when $i<t$.
        The envy from agent $j$ to agent $i$ is then equivalent to 
        $1+x_{j,2}\cdot v_{j,2}<x_{i,1}.$
        We have that $x_{j,2}\ge q_2$
        from definition of $X$, and $v_{j,2}\ge v_{t,2}$ since $j>t$.
        Combining the inequalities, obtain
        \begin{equation}\label{eq:claim-envy}
            q_2\cdot v_{t,2}<x_{i,1}-1.
        \end{equation}
        
        If $r_1> 0$, then $x_{i,1}\le p+q_1+1$.
        Then  \eqref{eq:claim-envy} gives $$q_2\cdot v_{t,2}< p+q_1+1-1,$$
        which is essentially the first part of the claim.
        If $r_1=0$, then $x_{i,1}=p+q_1$.
        Combining this, \eqref{eq:claim-envy} and  $v_{t,2}\cdot x_{t,2}>p-1$ gives
        $$q_2\cdot v_{t,2}<(p-1)+q_1<v_{t,2}\cdot x_{t,2}+q_1.$$
        The second part is also proved for case $i<t$.

        We move on to the remaining case $i=t$.
        The envy from agent $j$ to agent $t$ is then expressed as
        $1+x_{j,2}\cdot v_{j,2}<x_{t,1}+x_{t,2}\cdot v_{t,2}.$
        Similarly to the previous case, rewrite this as
        \begin{equation}\label{eq:claim-last}
            q_2\cdot v_{t,2}<(x_{t,1}-1)+x_{t,2}\cdot v_{t,2}.
        \end{equation}
        If $r_1>0$, $x_{t,1}=q_1+1$. Combining this, \eqref{eq:claim-last} and $p\ge x_{t,2}\cdot v_{t,2}$ gives the first part of the claim.
        If $r_1=0$, then $x_{t,1}=q_1$.
        Combining \eqref{eq:claim-last} with $x_{t,1}=q_1$ gives
        $$q_2\cdot v_{t,2}<q_1-1+x_{t,2}\cdot v_{t,2}<q_1+x_{t,2}\cdot v_{t,2}.$$\

        The proof of the case $i=t$ and the whole claim is complete.
    \end{claimproof}

    Clearly, \Cref{claim:from-re} and \Cref{claim:from-le} are controversial statements that are both true under the initial assumption.
    The obtained contradiction proves the theorem.
\end{proof}

While \Cref{thm:envy-one-direction} is fundamental to our split allocation approach, \Cref{thm:left-and-right-bounds} is pivotal to the binary search over $\mathcal{T}$.
The proof of \Cref{thm:left-and-right-bounds} is omitted due to the space constraints and can be found in the appendix.



\begin{proofO}
    By \Cref{thm:envy-one-direction} it is enough to prove that the $(t_L,k_L)$-split-allocation is not right-envious (first part), and that the $(t_R,k_R)$-split-allocation is not left-envious (second part).
    
    We start our proof with the second part of the theorem.
    Note that the maximal element in $\mathcal{T}$ with respect to $\prec$ is $(n,m_2)$ (by definition of $\mathcal{T}$).
    Denote by $X$ the $(n,m_2)$-split-distribuiton.
    Targeting towards a contradiction, assume that $X$ is left-envious.
    Then there are $i,j\in [n]$ with $i<j$ such that agent $i$ envies agent $j$ in $X$. 
    On the one hand, by \Cref{lem:no-envy-left-or-right}, we have that $i\le t$ and $j\ge t$.
    In our case $t=n$, consequently $j=n$ and $i<n$.
    On the other hand, by \Cref{lem:le-case-b} and \Cref{lem:envy-case-b-right} combined, there can be no envy between agent $i<n$ and $j=n$.
    This contradiction proves the second part of the theorem.

    We proceed to proving the first part of the theorem.
    For convenience, we denote the smallest element $(t_L,k_L)$ of $\mathcal{T}$ by just $(t,k)$.
    By definition of $\mathcal{T}$, $t=t_L$ is the smallest integer in $[n]$ such that
    \begin{itemize}
        \item $v_{t,2}\ge 1$, and
        \item $m_2\ge n-t+1$.
    \end{itemize}
    Then, $t$ is such that $m_2=n-t+1$ or $t-1\le n_1$ (since $n_1<n$ and $v_{n_1+1,2}\ge 1$ by definition of $n_1$).
    
    Then $k=k_L=\lfloor m_2/(n-t+1) \rfloor$.
    We denote the $(t,k)$-split-allocation by $X$.
    To prove the first part of the theorem, we assume that $X$ is right-envious, aiming to obtain a contradiction.

    By our assumption, there exist $i,j\in [n]$ such that $i<j$ and agent $i$ envies (up to any item) agent $j$ in $X$.
    By \Cref{lem:no-envy-left-or-right}, we have $i\le t$ and $j\ge t$.
    We consider several cases depending on whether $i=t$ or $j=t$, and in each case we come to a contradiction.

    \medskip\noindent\textbf{Agent $i<t$ envies agent $j=t$.}
    In this case, by \Cref{lem:envy-case-b-right}, we have that $m_1<p(t-1)$.
    Then by definition of split allocations, $x_{t,1}=0$.
    Then the envy (up to any item) from agent $i$ to agent $t$ is equivalent to
     $$x_{i,1}<(x_{t,2}-1)\cdot v_{i,2}.$$
     Consequently, by $v_{i,2}\le 1$, we have $x_{i,1}<x_{t,2}-1= k-1.$
     From definition of split allocations we know that $x_{s,1}\le x_{i,1}+1$ for each $s\in [t-1]$, and $m_1=\sum_{s\in[t-1]}x_{s,1}.$
     Then
     $$(k-1)(t-1)\ge m_1,$$
     so $k\ge 2$.
     This means that $m_2>n-t+1$, hence $t-1\le n_1$ holds.
    Finally obtain
     $$\frac{m_2}{n_2}=\frac{m_2}{n-n_1}\ge\frac{m_2}{n-t+1}\ge k> \frac{m_1}{t-1}\ge \frac{m_1}{n_1}.$$
     This contradicts the initial assumption $m_1/n_1\ge m_2/n_2$.
     The case $j=t$ is ruled out.

    In the remaining  two cases, we have $j>t$.
    The following claim is important for both cases.
    
    \begin{claimO}
        For each $j>t$, $x_{j,2}-1\le k$.
    \end{claimO}
    \begin{claimproofO}
  By definition of split allocations, we know that $x_{j,2}\le \lceil(m_2-k)/(n-t)\rceil.$

    Targeting towards a contradiction, suppose
    \begin{equation}\label{eq:divisible-nt}
    k < \left\lceil \frac{m_2-k}{n-t} \right\rceil - 1.
    \end{equation}

    If $m_2-k$ divisible by $n-t$, then $$k<\frac{m_2-k}{n-t}-1,$$
    and
    $$(n-t)\cdot k<(m_2-k)-(n-t),$$
    so
    $$(n-t+1)\cdot k+(n-t) < m_2.$$
    This contradicts the definition of $k$.

    Then $m_2-k$ has to be not divisible by $n-t$.
    Let $r>0$ be the remainder of division of $m_2-k$ by $n-t$.
    Then \eqref{eq:divisible-nt} becomes
    $$k<\left\lfloor \frac{m_2-k}{n-t}\right\rfloor=\frac{m_2-k-r}{n-t}.$$
    Multiply the left and the right by $(n-t)$ and obtain
    $$(n-t)\cdot k<m_2-k-r,$$
    so
    $$(n-t+1)\cdot k<m_2-r.$$
    Let $r'$ be the remainder of division of $m_2$ by $(n-t+1)$, then $(n-t+1)\cdot k+r'=m_2.$
    Consequently,
    $$m_2-r'<m_2-r,$$
    and $r'>r.$
    By definition of $r'$ and $r$ we have
    $$((n-t+1)\cdot k +r')-k=(n-t)\cdot q_2+r,$$
    so
    $$r'-r=(n-t)\cdot (q_2 -k).$$
    All parts are integers, so $r'-r$ is divisible by $n-t$.
    Consequently, $r'-r\ge n-t$.
    Since $r>0$, it has to be $r'\ge n-t+1$.
    This contradicts the definition of $r'$.
    The proof of the claim is complete.
    \end{claimproofO}

    We move on to the remaining two cases of envy between $i$ and $j$.

    \medskip\noindent\textbf{Agent $i<t$ envies agent $j>t$.}
    In this case, we have
    $$x_{i,1}<(x_{j,2}-1)\cdot v_{i,2}.$$
    Then $x_{i,1}<x_{j,2}-1\le k$,
    since $v_{i,2}\le 1$ by $i<t\le n_1+1$.
    Then note that $X$ is not case (b) of \Cref{def:split-distribution}, because in case (b) we would have $x_{i,1}\ge p=\lceil v_{t,2}\cdot k\rceil \ge k.$

    Then $X$ is case (a) of \Cref{def:split-distribution}.
    By definition, all items of the first type are distributed between agents in $[t-1]$, so
    $$\sum_{s=1}^{t-1}x_{s,1}= m_1.$$
    For each $s\in [t-1]$, $x_{s,1}\le x_{i,1}+1\le k$, while $x_{i,1}<k$.
    Consequently, $\sum_{s\in [t-1]}<k(t-1)$ and $k(t-1)>m_1$.
    Recall that it's either $t-1\le n_1$ or $m_2=n-t+1$.
    
    If $t-1\le n_1$, then, using $m_1/n_1\ge m_2/n_2$ and $n_1+n_2=n$,
    obtain
    $$k>\frac{m_1}{t-1}\ge \frac{m_1}{n_1}\ge \frac{m_2}{n_2}= \frac{m_2}{n-n_1}\ge \frac{m_2}{n-t+1}.$$
    This contradicts the definition of $k$.

    If $t-1>n_1$, then $m_2=n-t+1$.
    Then $k=1$ and $m_1<k(t-1)=t-1$.
    Consequently, $$m_1+m_2<(t-1)+(n-t+1)=n.$$
    This contradicts the initial setup.
    
    The analysis of the case $t\notin\{i,j\}$ is finished.
    
    \medskip\noindent\textbf{Agent $i=t$ envies agent $j>t$.}
    The envy is then expressed as
    $$x_{t,1}+x_{t,2}\cdot v_{t,2}<(x_{j,2}-1)\cdot v_{t,2},$$
    since agent $j$ receives no item of the first type in $X$.
    Consequently, $x_{t,2}<x_{j,2}-1\le k$.
    This contradicts the definition of $X$, since $x_{t,2}=k$.
    The case $i=t$ is also ruled out.

    The possible cases are exhausted, in each of the cases a contradiction is obtained.
    The proof of the theorem is complete.
\end{proofO}
\section{Reallocation}

\envyDirectionChange*

\begin{proof}

Let us denote by $Y$ the $(t',k')$-split-allocation. We will use $y_{i,1}$ and $y_{i,2}$ to denote the number of items of each type allocated to agent $i$ in $Y$.

We begin the proof with the right inequality.\\
Consider two cases: $t = t'$ and $t \neq t'$. \\
\textbf{Case 1:} $t = t'$. 

By the definition of the $(t,k)$-split-allocation, we have $k' = k - 1$. Assume the contrary:
\[ m_1 \ge \lceil x_{t+1,2}\, v_{t,2}\rceil\, t - \lceil (k - 1)\,v_{t,2}\rceil + 1. \]

By Lemma 3, envy can only occur between agents $i \leq t$ and $j \geq t$. Also, by the definition of a $(t, k)$-split-allocation: \\Since $x_{t+1,2} \ge x_{t,2} =k$, it follows that $\lceil x_{t+1,2}\, v_{t,2}\rceil \ge \lceil k\, v_{t,2}\rceil \ge \lceil k'\, v_{t,2}\rceil$. Therefore,
\[
m_1 \ge \lceil k'\, v_{t,2}\rceil\, (t - 1) + 1 = p\,(t - 1) + 1,
\]
where $p = \lceil k'\, v_{t,2}\rceil$.

From Lemma 5, it follows that there is no envy between  agents $i \leq t$ and $j \leq t$. However, by assumption, the $(t, k')$-allocation is right-envious, i.e., there exists an agent $i\le t$ who envies an agent $j>t$. Denote:
\[ m_1 - p\,(t - 1) = q_1\, t + r_1, \qquad 0 \le r_1 < t. \]

Then the number of items of the first type for agents is:
\[ 
y_{i,1} = 
\begin{cases}
q_1, & \text{if } i = t \text{ and } r_1 = 0,\\
q_1 + 1, & \text{if } i = t \text{ and } r_1 > 0,\\
p + q_1, & \text{if } i \le \min\{t - r_1, t - 1\},\\
p + q_1 + 1, & \text{if } i \in [t - r_1 + 1, t - 1]~.
\end{cases}
\]

Thus, for the agent $t$ we get:\\
\[ y_{t,1} = \left\lceil \frac{\,m_1 - p\,(t - 1)\,}{\,t\,} \right\rceil =\left\lceil \frac{\,m_1 - \lceil k'\, v_{t,2}\rceil\, (t - 1)\,}{\,t\,} \right\rceil. \]

By the contrary assumption, this is greater than or equal to:
\begin{multline*}
\left\lceil \frac{\,\lceil x_{t+1,2}\, v_{t,2}\rceil\, t - \lceil (k-1)\, v_{t,2}\rceil + 1 - \lceil (k-1)\, v_{t,2}\rceil\, (t - 1)\,}{\,t\,} \right\rceil 
\\
= \left\lceil \frac{\,\lceil x_{t+1,2}\, v_{t,2}\rceil\, t - \lceil (k-1)\, v_{t,2}\rceil\, t + 1\,}{\,t\,} \right\rceil 
\\
\ge \lceil x_{t+1,2}\, v_{t,2}\rceil - \lceil (k-1)\, v_{t,2}\rceil + 1.
\end{multline*}

For any agent $i < t$ we have:
\begin{multline*}
y_{i,1} \ge p + \left\lfloor \frac{\,m_1 - p\,(t - 1)\,}{\,t\,} \right\rfloor \ge p + (y_{t,1} - 1) 
\\
\ge p + (\lceil x_{t+1,2}\, v_{t,2}\rceil - \lceil (k-1)\, v_{t,2}\rceil) = \lceil x_{t+1,2}\, v_{t,2}\rceil~.
\end{multline*}

Note that for agents $j > t$, we have $y_{j,1} = 0$ and $y_{j,2} \leq x_{t+1,2} + 1$. Indeed, suppose for contradiction that there exists some agent $i > t$ such that $y_{i,2} \geq x_{t+1,2} + 2$.

From the properties of the $(t, k)$-split-allocation, we know that each agent $i > t$ receives at least $x_{t+1,2} + 1$ items of type 2. Therefore,
\[
m_2 \geq (x_{t+1,2} + 1)(n-t-1) + (x_{t+1,2} + 2) + (k-1),
\]
where $(k-1)$ accounts for the items allocated to agent $t$.

On the other hand, in the allocation $X$, we have $x_{i,2} \leq x_{t+1,2} + 1$ for all $i > t$. Summing over all such agents gives
\[
m_2 \leq k + (x_{t+1,2} + 1)(n-t-1) + x_{t+1,2}.
\]
This leads to a contradiction, since the lower bound on $m_2$ exceeds the upper bound. Therefore, $y_{j,2} \leq x_{t+1,2} + 1$ must hold for all $j > t$.
\\
\textit{Subcase 1.1:} Agent $t$ envies some $j >t$. 

Then 
\[ y_{t,1} + y_{t,2}\, v_{t,2} < (y_{j,2} - 1)\, v_{t,2} \le x_{t+1,2}\, v_{t,2}. \]
On the other hand:
\begin{multline*} y_{t,1} + y_{t,2}\, v_{t,2} \ge \\
(\lceil x_{t+1,2}\, v_{t,2}\rceil - \lceil (k-1)\, v_{t,2}\rceil + 1) + (k-1)\, v_{t,2}  >\\ 
\lceil x_{t+1,2}\, v_{t,2}\rceil - (\lceil (k-1)\, v_{t,2}\rceil -(k-1)\, v_{t,2} )+1 \ge \\ x_{t+1,2}\, v_{t_2,2}. \end{multline*}
Contradiction.\\
\textit{Subcase 1.2:} An agent $i <t$ envies some $j >t$.
Then 
\[ y_{i,1} < (y_{j,2} - 1)\, v_{i,2} \le x_{t+1,2}\, v_{t,2}. \]
We have $y_{i,1} \ge \lceil  x_{t+1,2}\, v_{t,2}\rceil$.  Then:
\[ \lceil x_{t+1,2}, v_{t_2,2}\rceil \le x_{i,1} < x_{t+1,2}\, v_{t_2,2}, \] 
which is impossible. \;Contradiction.\\
\textbf{Case 2:} $t \neq t'$.

\textit{Subcase 2.1: t' = n}\\
In this case, $k = 1$, $k' = n$, and $t = n-1$. It follows that $x_{t+1,2} = m_2 - 1$. We need to show that
\[
m_1 < \lceil (m_2 - 1)\, v_{t,2} \rceil\, (n-1).
\]
Suppose the contrary, that is,
\[
m_1 \geq \lceil (m_2 - 1)\, v_{t,2} \rceil\, (n-1).
\]
Now, consider the allocation $Y$. If there is any envy in $Y$, then $y_{n,1} = 0$; otherwise, by Lemma~5, if $m_1 \geq p(t-1)$, there can be no envy. Thus, $y_{n,1} = 0$, which implies that for any $i < n$,
\[
y_{i,1} \geq \lceil (m_2 - 1)\, v_{t,2} \rceil.
\]
But then agent $i$ cannot envy agent $n$, because otherwise
\[
\lceil (m_2 - 1)\, v_{t,2} \rceil \leq y_{i,1} < (y_{n,2} - 1)\, v_{i,2} \leq (m_2 - 1)\, v_{t,2},
\]
which is a contradiction. Therefore, this case is resolved.
\\
\textit{Subcase 2.2: $t' < n$}\\
In this case, by the definition of the $(t, k)$-ordering, $k = 1$ and $k' =\lfloor \frac{m_2}{n-t} \rfloor$ and $t' = t +1$. Let us again consider a possible envy situation. By assumption, there exist agents $i < j$ such that $i$ envies $j$. By Lemma~3, we have $i \le t'$ and $j \ge t'$. 

Now, observe that $i \neq t'$. Indeed, for any agent $j > t'$ we have
\[
y_{j,2} \le \left\lceil \frac{m_2}{n-t} \right\rceil \le \left\lfloor \frac{m_2}{n-t} \right\rfloor + 1 = y_{t',2} + 1,
\]
Therefore, agent $t'$ cannot be envious, and the envy must come from some agent $i < t'$.

Let us consider the maximum value of $y_{i,2}$ for $i > t$. We claim that
\[
\max_{i > t'} y_{i,2} \le x_{t+1,2} + 1.
\]
Indeed, suppose for contradiction that \[\max_{i > t'} y_{i,2} \ge x_{t+1,2} + 2.\] Then, for all agents $i \ge t'$, we would have $y_{i,2} \ge x_{t+1,2} + 1$. Thus,
\begin{multline*}
m_2 \ge (x_{t+1,2} + 1)\,(n-t' - 1) + (x_{t+1,2} + 2) =\\ (x_{t+1,2} + 1)\,(n-t') + 1 = (x_{t+1,2} + 1)\,(n-t - 1) + 1 .
\end{multline*}
On the other hand, for all $i > t$, it holds that $x_{i,2} \le x_{t+1,2} + 1$, and summing over all such agents yields
\begin{multline*}
m_2 \le (x_{t+1,2} + 1)\,(n-t-2) + x_{t+1,2} + x_{t,2} = \\
(x_{t+1,2} + 1)\,(n-t-2) + x_{t+1,2} + 1 = \\(x_{t+1,2} + 1)\,(n-t-1).
\end{multline*}
This is a contradiction, since the lower bound on $m_2$ exceeds the upper bound. Therefore, the maximum value of $y_{i,2}$ for $i > t'$ cannot exceed $x_{t+1,2} + 1$.

By the previous discussion, the only possible envy can arise from an agent $i < t'$ towards an agent $j \geq t'$. Note that if $i < t'$ envies $t'$, then by Lemma~5 it must hold that $m_1 < p (t'-1)$, where $p = \lceil k' v_{t',2}\rceil$. In this case, we have $y_{t',1} = 0$. Consequently, for any agent $j > t'$, it follows that $y_{j,2} \geq y_{t',2}$, and thus agent $i$ would also envy this agent $j > t'$.

Let us examine the implications of this. For $i < t'$, agent $i$ envies an agent $j > t'$, so:
\[
y_{i,1} < (y_{j,2} - 1) v_{i,2} \leq (y_{j,2} - 1) v_{t,2}.
\]
Thus, for all $i < t'$,
\[
y_{i,1} \leq (y_{j,2} - 1) v_{t,2}.
\]
Summing over all such agents, we obtain
\[
m_1 \leq (y_{j,2} - 1) v_{t,2} (t'-1).
\]

Now, upper-bounding $y_{j,2}$ by its maximum possible value $x_{t+1,2} +1$, we get
\[
m_1 \leq x_{t+1,2}\, v_{t,2}\, t.
\]

Note that we are considering the case $y_{t',1} = 0$, since otherwise $m_1 \geq p (t'-1)$, where $p = \lceil k' v_{t',2} \rceil$. In that case, for any $i < t'$, we have
\[
y_{i,1} \geq p \geq x_{t+1,2}\, v_{t',2},
\]
so agent $i$ cannot envy anyone, which contradicts the assumption of envy.

Therefore, the right inequality is established.

\medskip

Continue the proof with the left inequality.

Assume the opposite:
\[
m_1 \leq \lceil x_{t+1,2}\, v_{t,2}\rceil\, t - \lceil k\, v_{t,2}\rceil.
\]

From Lemma~4 and Lemma~5, under the presence of left-envious envy, agents $i \leq t$ and $j \leq t$ do not envy each other. Therefore, by Lemma 3 there must exist an agent $j > t$ who envies some agent $i \leq t$. Consider two possible cases:\\
\textbf{Case 1:} Agent $j > t$ envies an agent $i <t$.

Then:
\[ x_{t+1,2}\, v_{t,2} \le x_{t+1,2}, v_{j,2} \le x_{j,2}\, v_{j,2} < x_{i,1} - 1. \]

It follows that:
\[ x_{i,1} > x_{t+1,2}\, v_{t,2} + 1. \]

Since for all agents $i_1 <t$, the number of first-type items differs by at most one, we have 
\[
x_{i_1,1} \ge \lceil x_{t+1,2}\, v_{t,2} \rceil.
\]
Here we use the ceiling on the right-hand side, since $x_{i_1,1}$ is an integer.

In particular, for the agent $t$: 
\[ x_{t,1} > x_{t+1,2}\, v_{t,2} + 1 - \lceil k\, v_{t,2}\rceil. \]

Summing over all such agents, we obtain
\begin{multline*}
m_1 \ge \lceil x_{t+1,2}\, v_{t,2} \rceil (t-2) + (\lceil x_{t+1,2}\, v_{t,2} \rceil + 1) + \\ (\lceil x_{t+1,2}\, v_{t,2} \rceil + 1  - \lceil k\, v_{t,2}\rceil) = \lceil x_{t+1,2}\, v_{t,2} \rceil t  -  \lceil k\, v_{t,2}\rceil + 2,
\end{multline*}
which contradicts our assumption.\\
\textbf{Case 2:} Agent $j > t$ envies agent $t$.

In this case:
\[
x_{t+1,2}\, v_{t,2} \le x_{t+1,2}\, v_{j,2} \le x_{j,2}\, v_{j,2} < x_{t,1} - 1 + k\, v_{t,2}.
\]
Here, we subtract one because for any agent $j > t$ it holds that $v_{j,2} \ge 1$.

Then:
\[ x_{t,1} > ( x_{t+1,2} - k)\, v_{t,2} + 1. \]

From this, for any $i \in I1$:
\[ x_{i,1} > \Big\lceil (x_{t+1,2} - k)\, v_{t,2} + \lceil k\, v_{t,2}\rceil \Big\rceil \ge \lceil x_{t+1,2}\, v_{t,2}\rceil. \]

Summing over all such agents, we obtain
\begin{multline*}
m_1 \ge \lceil x_{t+1,2}\, v_{t,2} \rceil (t-1) +  \\ ( \lceil x_{t+1,2}\, v_{t,2} \rceil + 1  - \lceil k\, v_{t,2}\rceil) = \lceil x_{t+1,2}\, v_{t,2} \rceil t  -  \lceil k\, v_{t,2}\rceil + 1,
\end{multline*}
which contradicts our assumption.

Both possible cases lead to a contradiction, so the assumption was wrong. The left inequality is proved.

Thus, we have established both inequalities, and therefore the lemma is proved.

\end{proof}

We are now finally ready to proceed to the main theorem and complete the proof.
\reallocationefxpo*
\begin{proof}
Let us first prove Pareto optimality.
\begin{claim}
Allocation $X$ satisfies Pareto-optimality (PO).
\end{claim}
\begin{proof}
By Theorem~1, any proper allocation is Pareto optimal. We will show that $X$ is \emph{proper}. By the definition of a \emph{proper} allocation, it is enough to verify that
\begin{equation}\label{equat:le-case-b}
\max\{\,x_{i,2} : i > t\,\} < v_{t,2}.    
\end{equation}

Note that $x_{i,2} = 0$ for all $i < t$, so the condition above indeed captures all relevant cases.

First consider the case when $v_{t,2} \notin \mathbb{N}$. Assume the opposite: suppose there exists an agent $i > t$ such that 
\[ x_{i,2} \ge v_{t,2}. \]

From the construction of $X$, the values $x_{i,1}$ are determined as follows:
\[ 
x_{i,1} \le \left\lceil \frac{\,m_1 - \lceil d\, v_{t,2}\rceil\, t + \lceil k\, v_{t,2}\rceil\,}{\,1 + t\,} \right\rceil.
\]

Now, assuming $x_{i,2} \ge v_{t,2}$, we get:
\begin{equation}\label{equat:lebron}
\left\lceil \frac{\,m_1 - \lceil d\, v_{t,2}\rceil\, t + \lceil k\, v_{t,2}\rceil\,}{\,1 + t\,} \right\rceil \ge v_{t,2}.
\end{equation}

By Lemma 1, we have:
\[ m_1 \le \lceil d\, v_{t,2}\rceil\, t - \lceil (k - 1)\, v_{t,2}\rceil. \]

Substituting this into inequality see~\eqref{equat:lebron} and noting that $m_1$ is an integer, we obtain:
\begin{multline*}
\left\lceil \frac{\,\Big[\lceil d\, v_{t,2}\rceil\, t - \lceil (k - 1)\, v_{t,2}\rceil\Big] - \lceil d\, v_{t,2}\rceil\, t + \lceil k\, v_{t,2}\rceil\,}{\,1 + t\,} \right\rceil \ge v_{t,2},
\end{multline*}

which simplifies to 
\[
\left\lceil \frac{\,\lceil k\, v_{t,2}\rceil - \lceil (k - 1)\, v_{t,2}\rceil\,}{\,1 + t\,} \right\rceil \ge v_{t,2}.
\]

Now, estimate the numerator from above:
\[ \lceil k\, v_{t,2}\rceil - \lceil (k - 1)\, v_{t,2}\rceil < \lfloor v_{t,2} + 1 \rfloor . \]

Then

\[ \left\lceil \frac{\lfloor v_{t,2} + 1 \rfloor}{\,1 + t\,} \right\rceil > v_{t,2}. \]

Since $t \ge 1$, we have:

\[ \left\lceil \frac{\lfloor v_{t,2} + 1 \rfloor}{\,2\,} \right\rceil \ge v_{t,2}. \]

Then:

\[ \left\lceil \frac{ \lfloor v_{t,2} \rfloor  + 1}{2} \right\rceil \ge v_{t,2}. \]

Because $v_{t,2} \notin \mathbb{N}$. This implies 
\[ \left\lceil \frac{\lfloor v_{t,2} \rfloor + 1}{2} \right\rceil > \lfloor v_{t,2} \rfloor, \]

so 
\[ \frac{ \lfloor v_{t,2} \rfloor + 1}{2} > \lfloor v_{t,2} \rfloor, \] 
which in turn implies 
\[ \frac{\lfloor v_{t,2} \rfloor}{2} < \frac{1}{2}. \]

But this is impossible, since $v_{t,2} \ge 1$, hence $\lfloor v_{t,2} \rfloor \ge 1$, and then $\frac{ \lfloor v_{t,2 \rfloor}}{2} \ge \frac{1}{2}$ — a contradiction.

\medskip

Now consider the second case: $v_{t,2} \in \mathbb{N}$. Then:
\[ \lceil d\, v_{t,2}\rceil - \lceil k\, v_{t,2}\rceil = (d - k)\, v_{t,2}. \]

Therefore,
\[
x_{i,1} \le \left\lceil \frac{\,m_1 - d\, v_{t,2}\, t + k\, v_{t,2} - 1\,}{\,1 + t\,} \right\rceil.
\]
This is because, due to the equality
\[
\lceil d\, v_{t,2}\rceil - \lceil k\, v_{t,2}\rceil = (d - k)\, v_{t,2},
\]
the agents are prioritized:
\begin{itemize}
    \item Agents $1, 2, \ldots, t-1, t+1, \ldots, \ell, t$.
\end{itemize}

Substitute the bound on $m_1$ from Lemma 1:
\[ m_1 \le d\, v_{t,2}\, t - (k - 1)\, v_{t,2}. \]

Then:
\begin{multline*}
\left\lceil \frac{\,d\, v_{t,2}\, t - (k - 1)\, v_{t,2} - d\, v_{t,2}\, t + k\, v_{t,2} - 1\,}{\,1 + t\,} \right\rceil
=  \left\lceil \frac{\,v_{t,2} - 1\,}{\,1 + t\,} \right\rceil \ge v_{t,2}.
\end{multline*}

And since $t \ge 1$, it follows that
\[ \left\lceil \frac{\,v_{t,2} - 1\,}{\,2\,} \right\rceil \ge v_{t,2}. \]

This would imply 
\[ \frac{v_{t,2} - 1}{2} \ge v_{t,2} - \frac{1}{2}, \] 
which is a contradiction. 

Therefore, $X$ is a \emph{proper} allocation and, by Theorem 1, it is Pareto-optimal. Thus, the lemma is proved.
\end{proof}

Now let us prove that the allocation $X$ satisfies EFx. To this end, we will establish several statements about the absence of envy. Before we proceed, let us recall that $\ell \in [n]$ denotes the largest index such that $x_{\ell,2} = d$.

\begin{claim}
If there is envy between two agents $i$ and $j$, then $(i - t)(j - t) \leq 0$.
\end{claim}
\begin{proof}
Let us proceed by contradiction. Suppose that $(i-t)(j-t) > 0$.

\noindent\textbf{Case 1:} $i < t$, $j < t$.

These agents receive only items of the first type, and their quantities differ by at most $1$. Since they receive no other types of goods and the deviation in quantity is limited, envy between them is impossible.

\noindent\textbf{Case 2:} $i > t$, $j > t$.
Let us consider the possible cases based on the positions of the agents relative to agent $\ell$.

\noindent\textit{Case 2.1:}  $i > l$, $j > l$.
By construction of allocation $X$, all agents in this group receive identical sets of items. Therefore, envy is impossible.

\noindent\textit{Case 2.2:}  $i < l$, $j < l$.
Suppose, that agent $i$ envies agent $j$. Then
\[
x_{i,1} + x_{i,2} < x_{j,2} + x_{j,1} - \min(1,v_{i,2}).
\]
However, by the definition of the reallocation, $x_{i,2} = x_{j,2}$, and $x_{i,1}$ and $x_{j,1}$ differ by at most $1$. Therefore, envy is impossible in this case because $v_{i,2} > 1$.

\noindent\textit{Case 2.3:} Agent $j > l$ envies agent $i \le l$.
Then by the definition of envy:
\[
x_{j,2}\, v_{j,2} < (x_{i,1} - 1) + x_{i,2}\, v_{j,2}.
\]

Substituting $x_{j,2} = d+1$ and $x_{i,2} = d$, we get:
\[
(d+1)\, v_{j,2} < (x_{i,1} - 1) + d\, v_{j,2}. 
\]

Rearranging:

\[
(d+1)\, v_{j,2} - d\, v_{j,2} < x_{i,1} - 1, 
\]

So,
\[
v_{j,2} < x_{i,1} - 1.
\]

By condition~\eqref{equat:le-case-b}, we have
\[
x_{i,1} < v_{t,2}.
\]

But
\[
v_{t,2} \le v_{j,2},
\]

So, we have
\[
v_{j,2} < x_{i,1} - 1 <  x_{i,1} < v_{j,2}
\]

This is a contradiction.

\noindent\textit{Case 2.4:} Agent $i < l$ envies agent $j > l$.
Then:
\[
x_{i,2}\, v_{i,2} + x_{i,1} < (x_{j,2} - 1)\, v_{i,2}. 
\]

Substituting $x_{j,2} = d+1$ and $x_{i,2} = d$, we get:
\[
d\, v_{i,2} + x_{i,1} < ((d+1) - 1)\, v_{i,2}, 
\]

hence
\[
d\, v_{i,2} + x_{i,1} < d\, v_{i,2}. 
\]

Simplifying:
\[
x_{i,1} < 0, 
\]
a contradiction.

All cases have been considered and lead to contradictions. Therefore, no envy can occur between  $i$ and $j$, and the lemma is proved.
\end{proof}

\begin{claim}
There can be no envy between agents $i < t$ and agent $j > t$.
\end{claim}
\begin{proof}
Consider all possible cases of envy between an agent $i < t$ and an agent $j  > t$.

\noindent\textbf{Case 1:} Agent $j > l$ envies agent $i < t$.\\ 
Agent $j$ receives only items of the second type, and agent $i$ receives only items of the first type. If $j$ envied $i$, then:
\[ x_{j,2}\, v_{j,2} < x_{i,1} - 1. \]

However, by construction $x_{i,1} \le \lceil d\, v_{t,2}\rceil + v_{t,2}$ and $x_{j,2} = d+1$. Therefore,
\[ (d+1)\, v_{j,2} < \lceil d\, v_{t,2}\rceil + v_{t,2} - 1 \le (d+1)\, v_{t,2}. \]
And since $v_{j,2} \ge v_{t,2}$, this is impossible. Contradiction.

\noindent\textbf{Case 2:} Agent $t \ge j > t$ envies agent $i < t$.\\
Agent $j$ receives $x_{j,2} = d$ items of the second type and $x_{j,1} \ge 0$ items of the first type; agent $i$ receives only first-type items. If $j$ envied $i$, then:
\[ d\, v_{j,2} + x_{j,1} < x_{i,1} - 1. \]

However, by construction $x_{i,1} < \lceil d\, v_{t,2}\rceil + x_{j,1}$. Hence
\[ d\, v_{j,2} + x_{j,1} < \lceil d\, v_{t,2}\rceil + x_{j,1} - 1, \]

so
\[ d\, v_{j,2} < \lceil d\, v_{t,2}\rceil - 1. \]

But $v_{j,2} \ge v_{t,2}$, which implies $d\, v_{j,2} \ge d\, v_{t,2}$, whereas $\lceil d\, v_{t,2}\rceil - 1 < d\, v_{t,2}$ — again a contradiction.

\noindent\textbf{Case 3:} Agent $i < t $ envies agent $l \ge j >t$.\\
Agent $i$ receives only first-type items; agent $j$ receives $d$ items of the second type and $x_{j,1}$ items of the first type (with $x_{j,1} \le x_{i,1} - \lceil d\, v_{t,2}\rceil + 1$ by construction). Two subcases are possible:

\noindent\textit{Subcase 3.1: $1 \le  v_{i,2}$}\\ Then envy implies:
\[ x_{i,1} < d\, v_{i,2} + x_{j,1} - 1. \]
But $x_{j,1} \le x_{i,1} - \lceil d\, v_{t,2}\rceil + 1$, hence
\[ x_{i,1} < d\, v_{i,2} + x_{i,1} - \lceil d\, v_{t,2}\rceil, \]
so
\[ 0 < d\, v_{i,2} - \lceil d\, v_{t,2}\rceil, \]
which is impossible, since $v_{i,2} \le v_{t,2}$.

\noindent\textit{Subcase 3.2: $1 > v_{i,2}$}\\ Then envy implies:
\[ x_{i,1} < (d - 1)\, v_{i,2} + x_{j,1}. \]
And $x_{j,1} \le x_{i,1} - \lceil d\, v_{t,2}\rceil + 1$, therefore
\[ x_{i,1} < (d - 1)\, v_{i,2} + x_{i,1} - \lceil d\, v_{t,2}\rceil + 1, \]
which simplifies to
\[ 0 < (d - 1)\, v_{i,2} - \lceil d\, v_{t,2}\rceil + 1. \]
But since $v_{i,2} \le v_{t,2}$, this is impossible for $v_{t,2} \ge 1$.

\noindent\textbf{Case 4:} Agent $i < t$ envies agent $j > l$.
Agent $j$ receives only goods of the second type: $x_{j,2} = d+1$, $x_{j,1} = 0$. If $i$ envied $j$, then:
\[ x_{i,1} < (d+1 - 1)\, v_{i,2} = d\, v_{i,2}. \]

But by construction $x_{i,1} \ge \lceil d\, v_{t,2}\rceil$, and $v_{i,2} \le v_{t,2}$. Hence
\[ \lceil d\, v_{t,2}\rceil \le x_{i,1} < d\, v_{i,2}, \]
which is impossible, since the left side is a ceiling value.

Thus, in all cases envy is impossible, and the lemma is proved.
\end{proof}

\begin{claim}
There can be no envy between agent t and agent $l \ge i >t$.
\end{claim}
\begin{proof}
Let us consider two cases, depending on whether $t$ envies $i$ or $i$ envies $t$.

\noindent\textbf{Case 1:} Agent $i$ envies agent $t$.

Consider two subcases, depending on the relation between $\lceil d\, v_{t,2}\rceil - \lceil k\, v_{t,2}\rceil$ and $(d - k)\, v_{t,2}$:

\noindent\textit{Subcase 1.1} $\lceil d\, v_{t,2}\rceil - \lceil k\, v_{t,2}\rceil > (d - k)\, v_{t,2}$.
In this case, by the construction of the allocation:

\[ x_{i,1} \ge x_{t,1} - (\lceil d\, v_{t,2}\rceil - \lceil k\, v_{t,2}\rceil). \]
Because the agents are prioritized in the following order:
\begin{itemize}
    \item Agents $1, 2, \ldots, t-1,t, t+1, \ldots, \ell,$.
\end{itemize}

If envy were possible, then:
\begin{multline*}
x_{i,1} + d\, v_{i,2} = x_{i,1} + x_{i,2}\, v_{i,2} < x_{t,1} + x_{t,2}\, v_{i,2} - 1 = x_{t,1} - 1 + k\, v_{i,2},
\end{multline*}

which implies 
\[
x_{t,1} - (\lceil d\, v_{t,2}\rceil - \lceil k\, v_{t,2}\rceil) + d\, v_{i,2} < x_{t,1} - 1 + k\, v_{i,2},
\]
\[
\implies d\, v_{i,2} - (\lceil d\, v_{t,2}\rceil - \lceil k\, v_{t,2}\rceil) < k\, v_{i,2} - 1,
\] 
\[
\implies (d - k)\, v_{i,2} - (\lceil d\, v_{t,2}\rceil - \lceil k\, v_{t,2}\rceil) < -1,
\]
\[
\implies (d - k)\, v_{i,2} < (\lceil d\, v_{t,2}\rceil - \lceil k\, v_{t,2}\rceil) - 1,
\] 
\[
\implies (d - k)\, v_{i,2} < (d\, v_{t,2} + 1 - k\, v_{t,2}) - 1.
\]

This is a contradiction.

\noindent\textit{Subcase 1.2} $\lceil d\, v_{t,2}\rceil - \lceil k\, v_{t,2}\rceil < (d - k)\, v_{t,2}$.
Here, similarly:
\[ x_{i,1} \ge x_{t,1} - (\lceil d\, v_{t,2}\rceil - \lceil k\, v_{t,2}\rceil) - 1. \]

If envy were possible:
\[
x_{i,1} + d\, v_{i,2} < x_{t,1} - 1 + k\, v_{i,2},
\]

then 
\[
x_{t,1} - (\lceil d\, v_{t,2}\rceil - \lceil k\, v_{t,2}\rceil) - 1 + d\, v_{i,2} < x_{t,1} - 1 + k\, v_{i,2},
\] 

which simplifies to 
\[
(d - k)\, v_{i,2} < (\lceil d\, v_{t,2}\rceil - \lceil k\, v_{t,2}\rceil).
\]

But by assumption of this subcase:
\[ \lceil d\, v_{t,2}\rceil - \lceil k\, v_{t,2}\rceil > (d - k)\, v_{t,2}. \]

This is a contradiction.

\noindent\textbf{Case 2:} Agent $t$ envies agent $i$.
Consider two subcases, depending on the sign of $\lceil d\, v_{t,2}\rceil - \lceil k\, v_{t,2}\rceil$ versus $(d - k)\, v_{t,2}$:

\noindent\textit{Subcase 2.1} $\lceil d\, v_{t,2}\rceil - \lceil k\, v_{t,2}\rceil > (d - k)\, v_{t,2}$.
In this case, by construction:
\[ x_{t,1} \ge x_{i,1} + (\lceil d\, v_{t,2}\rceil - \lceil k\, v_{t,2}\rceil) - 1. \]

If envy were possible:
\[
x_{i,1} + d\, v_{t,2} - 1 > x_{t,1} + k\, v_{t,2},
\] 

then substituting the bound for $x_{t,1}$:
\[
x_{i,1} + d\, v_{t,2} - 1 > x_{i,1} + (\lceil d\, v_{t,2}\rceil - \lceil k\, v_{t,2}\rceil) - 1 + k\, v_{t,2},
\] 
which simplifies to
\[
d\, v_{t,2} > (\lceil d\, v_{t,2}\rceil - \lceil k\, v_{t,2}\rceil) + k\, v_{t,2},
\]

\[
d\, v_{t,2} - k\, v_{t,2} > \lceil d\, v_{t,2}\rceil - \lceil k\, v_{t,2}\rceil,
\] 
\[
(d - k)\, v_{t,2} > \lceil d\, v_{t,2}\rceil - \lceil k\, v_{t,2}\rceil.
\]

But by assumption of this subcase:
\[ \lceil d\, v_{t,2}\rceil - \lceil k\, v_{t,2}\rceil > (d - k)\, v_{t,2}. \]

This is a contradiction.

\noindent\textit{Subcase 2.2} $\lceil d\, v_{t,2}\rceil - \lceil k\, v_{t,2}\rceil < (d - k)\, v_{t,2}$.
In this case, by construction:
\[ x_{t,1} \ge x_{i,1} + (\lceil d\, v_{t,2}\rceil - \lceil k\, v_{t,2}\rceil). \]

Because the agents are prioritized in the following order: agents $1, 2, \ldots, t-1,t, t+1, \ldots, \ell$.

If envy were possible:
\[
x_{i,1} + d\, v_{t,2} - 1 > x_{t,1} + k\, v_{t,2},
\] 

then substituting the bound for $x_{t,1}$:
\[
x_{i,1} + d\, v_{t,2} - 1 > x_{i,1} + (\lceil d\, v_{t,2}\rceil - \lceil k\, v_{t,2}\rceil) + k\, v_{t,2},
\]

which simplifies to
\[
d\, v_{t,2} - 1 > (\lceil d\, v_{t,2}\rceil - \lceil k, v_{t,2}\rceil) + k\, v_{t,2},
\]
\[
(d - k)\, v_{t,2} - 1 > \lceil d\, v_{t,2}\rceil - \lceil k\, v_{t,2}\rceil,
\] 
\[
(d - k)\, v_{t,2} - 1 > d\, v_{t,2} - (k\, v_{t,2} + 1).
\]

This is a contradiction.

Thus, envy is impossible in this case as well. The lemma is proved.
\end{proof}

\begin{claim}
There can be no envy between agent t and agent $i >l$.
\end{claim}
\begin{proof}
Let us consider two cases, depending on whether $t$ envies $i$ or $i$ envies $t$.\\
\textbf{Case 1:} Agent $t$ envies agent $i$.

By Lemma 1, 
\[\lceil x_{t+1,2}\cdot v_{t,2}\rceil t-\lceil kv_{t,2}\rceil + 1\le m_1 \]

Then we know that
\[ x_{t,1} \ge \lceil d\, v_{t,2}\rceil - \lceil k\, v_{t,2}\rceil + 1. \]

If envy exists, then:
\[ x_{t,1} + x_{t,2}\, v_{t,2} < (x_{i,2} - 1)\, v_{t,2} = d\, v_{t,2}. \]

But $x_{t,2} = k$ and $x_{i,2} = d+1$. Substituting:
\[ \lceil d\, v_{t,2}\rceil - \lceil k\, v_{t,2}\rceil + 1 + k\, v_{t,2} < d\, v_{t,2}. \]

Bringing everything to one side:
\[ -\,\lceil k\, v_{t,2}\rceil + 1 + k\, v_{t,2} < 0. \]

But $k\, v_{t,2} + 1 \ge \lceil k\, v_{t,2}\rceil$. Contradiction.
Thus, an agent $t$ cannot envy an agent $i > l$.

\noindent\textbf{Case 2:} Agent $i > l$ envies agent $t$.

Consider two subcases, depending on the relation between $\lceil d\, v_{t,2}\rceil - \lceil k\, v_{t,2}\rceil$ and $(d - k)\, v_{t,2}$:

\noindent\textit{Subcase 2.1} $\lceil d\, v_{t,2}\rceil - \lceil k\, v_{t,2}\rceil > (d - k)\, v_{t,2}$.

In this case, by construction:
\[ x_{t,1} \le x_{j,1} + (\lceil d\, v_{t,2}\rceil - \lceil k\, v_{t,2}\rceil) < v_{t,2} + (\lceil d\, v_{t,2}\rceil - \lceil k\, v_{t,2}\rceil), \]

where $j$ is an agent $l \ge j>t$ and by ~\eqref{equat:le-case-b}, $x_{j,1} < v_{t,2}$.

If envy were possible, then:
\begin{multline*}
(d+1)\, v_{i,2} = x_{i,2}\, v_{i,2} < x_{t,1} - 1 + x_{t,2}\, v_{i,2} \le  v_{t,2} + (\lceil d\, v_{t,2}\rceil - \lceil k\, v_{t,2}\rceil) + k\, v_{t,2} - 1.
\end{multline*}

Note that $k\, v_{t,2} + v_{t,2} = (k + 1)\, v_{t,2}$. Therefore:
\begin{multline*}
v_{t,2} + (\lceil d\, v_{t,2}\rceil - \lceil k\, v_{t,2}\rceil) + k\, v_{t,2} - 1 =
\lceil d\, v_{t,2}\rceil + (k\, v_{t,2} - \lceil k\, v_{t,2}\rceil) + v_{t,2} - 1. 
\end{multline*}

But $k\, v_{t,2} - \lceil k\, v_{t,2}\rceil \le 0$. Hence
\[ x_{i,2}\, v_{i,2} < \lceil d\, v_{t,2}\rceil + v_{t,2} - 1. \]

By definition $x_{i,2} = d+1$, and since $v_{i,2} \le v_{t,2}$, we have
\[ (d+1)\, v_{i,2} \le (d+1)\, v_{t,2}. \]

Thus,
\[ (d+1)\, v_{i,2} < \lceil d\, v_{t,2}\rceil + v_{t,2} - 1 \le (d+1)\, v_{t,2}, \]
which is impossible, as the left side cannot be strictly less than the right side for all $v_{i,2} \le v_{t,2}$. Contradiction.

\noindent\textit{Subcase 2.2} $\lceil d\, v_{t,2}\rceil - \lceil k\, v_{t,2}\rceil \le (d - k)\, v_{t,2}$.

Then by our construction:
\[ x_{t,1} \le x_{j,1} + 1 + (\lceil d\, v_{t,2}\rceil - \lceil k\, v_{t,2}\rceil) < v_{t,2} + 1 + (d - k)\, v_{t,2}, \]
where $j$ is an agent $l \ge j>t$, and by ~\eqref{equat:le-case-b}, $x_{j,1} < v_{t,2}$.

If envy were possible, then:
\[
(d+1)\, v_{i,2} = x_{i,2}\, v_{i,2} < x_{t,1} - 1 + x_{t,2}\, v_{i,2} \le v_{t,2} + 1 + (d - k)\, v_{t,2} + k\, v_{i,2} - 1 = (d+1)\, v_{i,2}.
\]

That is,
\[ (d+1)\, v_{i,2} < (d+1)\, v_{i,2}, \]
which is impossible.

Thus, all cases have been considered, and the lemma is proved.
\end{proof}

We have already shown that the allocation $X$ satisfies Pareto optimality by Claim~7. Moreover, possible cases of envy have been considered in Claims~8, 9, 10, and 11. It remains only to prove that there can be no envy between agent $t$ and any agent $i < t$. Once this is established, the theorem will be proved.

Note that, since $m_1 > p (t-1)$, and since in the allocation of the first type of items, agent $t$ is always prioritized above all agents $i < t$, the proof of absence of envy in this case proceeds analogously to Lemma~5.

Therefore, the theorem is proved.

\end{proof}

\section{Conclusion}

 In this work, we have made progress in understanding the existence and computability of EFX+PO allocations, establishing that such allocations always exist for two types of goods with positive utilities and developing an efficient $\mathcal{O}(n \log n + \log m)$-time algorithm. We are interested in whether these results can be extended to several simple independent cases of a more general open question.
 It represents critical next steps in the broader research aimed at characterizing when EFX and EFX+PO allocations are guaranteed to exist.

 \begin{question}
    Does an EFX allocation always exist? Does an EFX+PO allocation always exists when all utilities are positive? At least in two settings:
    \begin{enumerate}
        \item[(a)] utility matrix has rank $2$,
        \item[(b)] three types of goods.
    \end{enumerate}
\end{question}

\bibliography{division}

@article{ChaudhuryGM24_3agent_EFX,
  author       = {Bhaskar Ray Chaudhury and
                  Jugal Garg and
                  Kurt Mehlhorn},
  title        = {{EFX} Exists for Three Agents},
  journal      = {J. {ACM}},
  volume       = {71},
  number       = {1},
  pages        = {4:1--4:27},
  year         = {2024},
  url          = {https://doi.org/10.1145/3616009},
  doi          = {10.1145/3616009},
  bibsource    = {dblp computer science bibliography, https://dblp.org}
}

@inproceedings{lipton,
author = {Lipton, R. J. and Markakis, E. and Mossel, E. and Saberi, A.},
title = {On Approximately Fair Allocations of Indivisible Goods},
year = {2004},
isbn = {1581137710},
publisher = {Association for Computing Machinery},
address = {New York, NY, USA},
url = {https://doi.org/10.1145/988772.988792},
doi = {10.1145/988772.988792},
booktitle = {Proceedings of the 5th ACM Conference on Electronic Commerce},
pages = {125–131},
numpages = {7},
location = {New York, NY, USA},
series = {EC '04}
}

@article{aziz2022algorithmic,
  author       = {Haris Aziz and
                  Bo Li and
                  Herv{\'{e}} Moulin and
                  Xiaowei Wu},
  title        = {Algorithmic fair allocation of indivisible items: a survey and new
                  questions},
  journal      = {SIGecom Exch.},
  volume       = {20},
  number       = {1},
  pages        = {24--40},
  year         = {2022},
  url          = {https://doi.org/10.1145/3572885.3572887},
  doi          = {10.1145/3572885.3572887},
  bibsource    = {dblp computer science bibliography, https://dblp.org}
}

@article{PlautR20,
  author       = {Benjamin Plaut and
                  Tim Roughgarden},
  title        = {Almost Envy-Freeness with General Valuations},
  journal      = {{SIAM} J. Discret. Math.},
  volume       = {34},
  number       = {2},
  pages        = {1039--1068},
  year         = {2020},
  url          = {https://doi.org/10.1137/19M124397X},
  doi          = {10.1137/19M124397X},
  bibsource    = {dblp computer science bibliography, https://dblp.org}
}

@article{overview23,
title = {Fair division of indivisible goods: Recent progress and open questions},
journal = {Artificial Intelligence},
volume = {322},
pages = {103965},
year = {2023},
issn = {0004-3702},
doi = {https://doi.org/10.1016/j.artint.2023.103965},
url = {https://www.sciencedirect.com/science/article/pii/S000437022300111X},
author = {Georgios Amanatidis and Haris Aziz and Georgios Birmpas and Aris Filos-Ratsikas and Bo Li and Hervé Moulin and Alexandros A. Voudouris and Xiaowei Wu}
}

@ARTICLE{Budish,
title = {The Combinatorial Assignment Problem: Approximate Competitive Equilibrium from Equal Incomes},
author = {Budish, Eric},
year = {2011},
journal = {Journal of Political Economy},
volume = {119},
number = {6},
pages = {1061 - 1103},
url = {https://EconPapers.repec.org/RePEc:ucp:jpolec:doi:10.1086/664613}
}

@article{Caragiannis_2019,
author = {Caragiannis, Ioannis and Kurokawa, David and Moulin, Hervé and Procaccia, Ariel and Shah, Nisarg and Wang, Junxing},
year = {2019},
month = {09},
pages = {1-32},
title = {The Unreasonable Fairness of Maximum Nash Welfare},
volume = {7},
journal = {ACM Transactions on Economics and Computation},
doi = {10.1145/3355902}
}

@inproceedings{two_types_goods,
  author       = {Pranay Gorantla and
                  Kunal Marwaha and
                  Santhoshini Velusamy},
  editor       = {Nikhil Bansal and
                  Viswanath Nagarajan},
  title        = {Fair allocation of a multiset of indivisible items},
  booktitle    = {Proceedings of the 2023 {ACM-SIAM} Symposium on Discrete Algorithms,
                  {SODA} 2023, Florence, Italy, January 22-25, 2023},
  pages        = {304--331},
  publisher    = {{SIAM}},
  year         = {2023},
  url          = {https://doi.org/10.1137/1.9781611977554.ch13},
  doi          = {10.1137/1.9781611977554.CH13},
  bibsource    = {dblp computer science bibliography, https://dblp.org}
}

@article{two_types_agents,
title = {Existence of EFX for two additive valuations},
journal = {Discrete Applied Mathematics},
volume = {340},
pages = {115-122},
year = {2023},
issn = {0166-218X},
doi = {https://doi.org/10.1016/j.dam.2023.06.035},
url = {https://www.sciencedirect.com/science/article/pii/S0166218X2300255X},
author = {Ryoga Mahara}
}

@article{mnw_efx,
title = {Maximum Nash welfare and other stories about EFX},
journal = {Theoretical Computer Science},
volume = {863},
pages = {69-85},
year = {2021},
issn = {0304-3975},
doi = {https://doi.org/10.1016/j.tcs.2021.02.020},
url = {https://www.sciencedirect.com/science/article/pii/S0304397521000931},
author = {Georgios Amanatidis and Georgios Birmpas and Aris Filos-Ratsikas and Alexandros Hollender and Alexandros A. Voudouris}
}

@inproceedings{complexity_results_survey,
  author       = {Trung Thanh Nguyen and
                  J{\"{o}}rg Rothe},
  title        = {Complexity Results and Exact Algorithms for Fair Division of Indivisible
                  Items: {A} Survey},
  booktitle    = {Proceedings of the Thirty-Second International Joint Conference on
                  Artificial Intelligence, {IJCAI} 2023, 19th-25th August 2023, Macao,
                  SAR, China},
  pages        = {6732--6740},
  publisher    = {ijcai.org},
  year         = {2023},
  url          = {https://doi.org/10.24963/ijcai.2023/754},
  doi          = {10.24963/IJCAI.2023/754},
  bibsource    = {dblp computer science bibliography, https://dblp.org}
}

@inproceedings{Gravin_donating,
  author       = {Ioannis Caragiannis and
                  Nick Gravin and
                  Xin Huang},
  editor       = {Anna R. Karlin and
                  Nicole Immorlica and
                  Ramesh Johari},
  title        = {Envy-Freeness Up to Any Item with High Nash Welfare: The Virtue of
                  Donating Items},
  booktitle    = {Proceedings of the 2019 {ACM} Conference on Economics and Computation,
                  {EC} 2019, Phoenix, AZ, USA, June 24-28, 2019},
  pages        = {527--545},
  publisher    = {{ACM}},
  year         = {2019},
  url          = {https://doi.org/10.1145/3328526.3329574},
  doi          = {10.1145/3328526.3329574},
  bibsource    = {dblp computer science bibliography, https://dblp.org}
}

@article{ChaudhuryKMS21,
  author       = {Bhaskar Ray Chaudhury and
                  Telikepalli Kavitha and
                  Kurt Mehlhorn and
                  Alkmini Sgouritsa},
  title        = {A Little Charity Guarantees Almost Envy-Freeness},
  journal      = {{SIAM} J. Comput.},
  volume       = {50},
  number       = {4},
  pages        = {1336--1358},
  year         = {2021},
  url          = {https://doi.org/10.1137/20M1359134},
  doi          = {10.1137/20M1359134},
  bibsource    = {dblp computer science bibliography, https://dblp.org}
}

@inproceedings{ChaudhuryGMMM21,
  author       = {Bhaskar Ray Chaudhury and
                  Jugal Garg and
                  Kurt Mehlhorn and
                  Ruta Mehta and
                  Pranabendu Misra},
  editor       = {P{\'{e}}ter Bir{\'{o}} and
                  Shuchi Chawla and
                  Federico Echenique},
  title        = {Improving {EFX} Guarantees through Rainbow Cycle Number},
  booktitle    = {{EC} '21: The 22nd {ACM} Conference on Economics and Computation,
                  Budapest, Hungary, July 18-23, 2021},
  pages        = {310--311},
  publisher    = {{ACM}},
  year         = {2021},
  url          = {https://doi.org/10.1145/3465456.3467605},
  doi          = {10.1145/3465456.3467605},
  bibsource    = {dblp computer science bibliography, https://dblp.org}
}

@inproceedings{BergerCFF22,
  author       = {Ben Berger and
                  Avi Cohen and
                  Michal Feldman and
                  Amos Fiat},
  title        = {Almost Full {EFX} Exists for Four Agents},
  booktitle    = {Thirty-Sixth {AAAI} Conference on Artificial Intelligence, {AAAI}
                  2022, Thirty-Fourth Conference on Innovative Applications of Artificial
                  Intelligence, {IAAI} 2022, The Twelveth Symposium on Educational Advances
                  in Artificial Intelligence, {EAAI} 2022 Virtual Event, February 22
                  - March 1, 2022},
  pages        = {4826--4833},
  publisher    = {{AAAI} Press},
  year         = {2022},
  url          = {https://doi.org/10.1609/aaai.v36i5.20410},
  doi          = {10.1609/AAAI.V36I5.20410},
  bibsource    = {dblp computer science bibliography, https://dblp.org}
}

@article{Mahara24,
  author       = {Ryoga Mahara},
  title        = {Extension of Additive Valuations to General Valuations on the Existence
                  of {EFX}},
  journal      = {Math. Oper. Res.},
  volume       = {49},
  number       = {2},
  pages        = {1263--1277},
  year         = {2024},
  url          = {https://doi.org/10.1287/moor.2022.0044},
  doi          = {10.1287/MOOR.2022.0044},
  bibsource    = {dblp computer science bibliography, https://dblp.org}
}

@inproceedings{GhosalHN025,
  author       = {Pratik Ghosal and
                  Vishwa Prakash HV and
                  Prajakta Nimbhorkar and
                  Nithin Varma},
  editor       = {Toby Walsh and
                  Julie Shah and
                  Zico Kolter},
  title        = {(Almost Full) {EFX} for Three (and More) Types of Agents},
  booktitle    = {AAAI-25, Sponsored by the Association for the Advancement of Artificial
                  Intelligence, February 25 - March 4, 2025, Philadelphia, PA, {USA}},
  pages        = {13889--13896},
  publisher    = {{AAAI} Press},
  year         = {2025},
  url          = {https://doi.org/10.1609/aaai.v39i13.33519},
  doi          = {10.1609/AAAI.V39I13.33519},
  bibsource    = {dblp computer science bibliography, https://dblp.org}
}

@article{AmanatidisMN20,
  author       = {Georgios Amanatidis and
                  Evangelos Markakis and
                  Apostolos Ntokos},
  title        = {Multiple birds with one stone: Beating 1/2 for {EFX} and {GMMS} via
                  envy cycle elimination},
  journal      = {Theor. Comput. Sci.},
  volume       = {841},
  pages        = {94--109},
  year         = {2020},
  url          = {https://doi.org/10.1016/j.tcs.2020.07.006},
  doi          = {10.1016/J.TCS.2020.07.006},
  bibsource    = {dblp computer science bibliography, https://dblp.org}
}

@inproceedings{AmanatidisFS24,
  author       = {Georgios Amanatidis and
                  Aris Filos{-}Ratsikas and
                  Alkmini Sgouritsa},
  editor       = {Dirk Bergemann and
                  Robert Kleinberg and
                  Daniela Sab{\'{a}}n},
  title        = {Pushing the Frontier on Approximate {EFX} Allocations},
  booktitle    = {Proceedings of the 25th {ACM} Conference on Economics and Computation,
                  {EC} 2024, New Haven, CT, USA, July 8-11, 2024},
  pages        = {1268--1286},
  publisher    = {{ACM}},
  year         = {2024},
  url          = {https://doi.org/10.1145/3670865.3673582},
  doi          = {10.1145/3670865.3673582},
  bibsource    = {dblp computer science bibliography, https://dblp.org}
}

@book{EF_foley1966resource,
  title={Resource allocation and the public sector},
  author={Foley, Duncan Karl},
  year={1966},
  publisher={Yale University}
}

@inproceedings{three_types_agents_EC,
  author       = {Vishwa {Prakash HV} and
                  Pratik Ghosal and
                  Prajakta Nimbhorkar and
                  Nithin Varma},
  editor       = {Itai Ashlagi and
                  Aaron Roth},
  title        = {{EFX} Exists for Three Types of Agents},
  booktitle    = {Proceedings of the 26th {ACM} Conference on Economics and Computation,
                  {EC} 2025, Stanford University, Stanford, CA, USA, July 7-10, 2025},
  pages        = {101--128},
  publisher    = {{ACM}},
  year         = {2025},
  url          = {https://doi.org/10.1145/3736252.3742509},
  doi          = {10.1145/3736252.3742509},
  bibsource    = {dblp computer science bibliography, https://dblp.org}
}

@article{EFX_bivalued,
title = {Computing fair and efficient allocations with few utility values},
journal = {Theoretical Computer Science},
volume = {962},
pages = {113932},
year = {2023},
issn = {0304-3975},
doi = {https://doi.org/10.1016/j.tcs.2023.113932},
url = {https://www.sciencedirect.com/science/article/pii/S0304397523002451},
author = {Jugal Garg and Aniket Murhekar}
}

@inproceedings{EFX_lexicographic,
  author       = {Hadi Hosseini and
                  Sujoy Sikdar and
                  Rohit Vaish and
                  Lirong Xia},
  title        = {Fair and Efficient Allocations under Lexicographic Preferences},
  booktitle    = {Thirty-Fifth {AAAI} Conference on Artificial Intelligence, {AAAI}
                  2021, Thirty-Third Conference on Innovative Applications of Artificial
                  Intelligence, {IAAI} 2021, The Eleventh Symposium on Educational Advances
                  in Artificial Intelligence, {EAAI} 2021, Virtual Event, February 2-9,
                  2021},
  pages        = {5472--5480},
  publisher    = {{AAAI} Press},
  year         = {2021},
  url          = {https://doi.org/10.1609/aaai.v35i6.16689},
  doi          = {10.1609/AAAI.V35I6.16689},
  bibsource    = {dblp computer science bibliography, https://dblp.org}
}
\bibliographystyle{apalike}

\end{document}